\documentclass[12pt]{article}
\usepackage[margin=1in]{geometry}
\usepackage[round,compress]{natbib}
\usepackage{setspace}
\usepackage{xcolor}
\usepackage{hyperref}
\usepackage{amsmath,amsthm,amssymb}
\usepackage{enumitem}
\newtheorem{theorem}{Theorem}

\usepackage{rotating,multirow,makecell,array}
\usepackage {tikz}
\usepackage{textcomp}
\usetikzlibrary{calc} 
\usepackage{subcaption}
\usetikzlibrary {positioning}
  \newtheorem{cor}{Corollary}
    \newtheorem{lem}{Lemma}
  \newtheorem{prop}{Proposition}
\theoremstyle{definition}
 \newtheorem{defn}{Definition}
 \newtheorem{setting}{Setting}
  \usepackage{multibib}
\usepackage[utf8]{inputenc}   
\usepackage[T1]{fontenc}      
\usepackage{graphicx}         
\usepackage{booktabs}         
\usepackage{epstopdf}         
\usepackage{booktabs}   
\usepackage{multirow}   
\usepackage{booktabs}
\usepackage{multirow}      
\usepackage{makecell}      
 \theoremstyle{remark}
 \newtheorem{exmp}{Example}
 \newtheorem{rem}{Remark}
 
\usepackage{booktabs}   
\usepackage{amsmath}    

\usepackage{caption}              
\ifxetex
  \usepackage{mathspec}
\fi

\DeclareMathOperator*\argmin{arg\,min}
\DeclareMathOperator*\argmax{arg\,max}

\title{Publication Design with Incentives in Mind%
\footnote{The first version of this paper was circulated in October, 2024. We thank Isaiah Andrews, Dmitry Arkhangelesky, Arun Chandrasekhar, Kevin Chen, Raj Chetty, David Cutler, Matt Gentzkow, Ed Glaeser, Matt Jackson, Max Kasy, Toru Kitagawa, Paul Niehaus, Marco Ottaviani, Anna Russo, Jesse Shapiro, Jann Spiess, Elie Tamer, Alex Tetenov, and Kaspar W\"uthrich for helpful discussion. We thank Karthik Seetharaman for exceptional research assistance. The usual disclaimer applies.  }
}
\author{Ravi Jagadeesan%
\footnote{Department of Economics, Stanford University.  Email address: \texttt{ravi.jagadeesan@gmail.com}.}
\and Davide Viviano%
\footnote{Department of Economics, Harvard University.  Email address: \texttt{dviviano@fas.harvard.edu}.}
}
\date{\today}

\begin{document}

\maketitle

\vspace{-24pt}

\begin{abstract}

The publication process both determines which research receives the most attention, and influences the supply of research through its impact on researchers' private incentives. We introduce a framework to study optimal publication decisions when researchers can choose (i) whether or how to conduct a study and (ii) whether or how to manipulate the research findings (e.g., via selective reporting or data manipulation). When manipulation is not possible, but research entails substantial private costs for the researchers, it may be optimal to incentivize cheaper research designs even if they are less accurate. When manipulation is possible, it is optimal to publish some manipulated results, as well as results that would have not received attention in the absence of manipulability. Even if it is possible to deter manipulation, such as by requiring pre-registered experiments instead of (potentially manipulable) observational studies, it is suboptimal to do so when experiments entail high research costs.
We illustrate the implications of our model in an application to medical studies.

\end{abstract}


\onehalfspacing 

\section{Introduction}

Publication decisions shape the process of scientific communication. By selecting what to publish, journals affect which findings receive the most attention and can inform the public about the state of the world. 
The design of publication rules has therefore motivated recent debates on how statistical significance should affect publication when the goal is to direct attention to the most informative results \citep[e.g.,][]{abadie2020statistical,frankel2022findings}.

However, the publication process also affects the supply of research by influencing  researchers' incentives about how to conduct research. Researchers have many degrees of freedom about how to conduct their research, such as how and where to run an experiment \citep[e.g.,][]{allcott2015site,gechter2022combining}; the size, cost, and effort associated with the study \citep[e.g.,][]{thompson2000should, grabowski2002returns}; and which findings to report from a given study \citep[][]{brodeur2020methods,elliott2022detecting}.
We refer broadly to their choices about each of these aspects of a study broadly as a \emph{research design}.

Researcher's private incentives may influence how they choose their design.
Yet,
``while economists assiduously apply incentive theory to the outside world, we use research methods that rely on the assumption that social scientists are saintly automatons'' \citep{glaeser2006researcher}. 
This raises the questions of when and how researchers' incentives should impact the design of publication processes,
and, more broadly, the optimal allocation of attention to research. 

This paper studies optimal publication decisions when researchers choose research designs based on private costs and benefits. 
We frame this question as a mechanism design problem: a social planner (principal) optimizes a publication rule, taking into account the incentives of a researcher (agent). 
The planner aims to use the publication process to efficiently allocate the attention of the audience to research findings.

More specifically, as in \cite{frankel2022findings} (building in turn on \cite{wald1950statistical}),
we suppose that research results impact the actions of an audience
who has limits on how much attention they can devote to research. The social planner seeks to publish results that are most important for the audience, net of the cost of (or taking into account constraints on) publication or attention.  
Due to attention costs, not all results will be published,
which leaves the planner with a non-trivial trade-off about which results and designs to publish.
We then introduce a model in which publication decisions affect the supply of research in the first place (Section~\ref{sec:setup}): given the publication rule, the researcher chooses the design that maximizes her value from publication (or other attention) net of research costs.

As a concrete example, consider a medical journal that is seeking to decide whether to publish results from a clinical study. 
The journal wants to convey accurate information on drug efficacy and direct the audience's attention to the most effective drugs.%
\footnote{%
For example, the stated mission of the \emph{New England Journal of Medicine} is
``to publish the best research and information at the intersection of biomedical science and clinical practice and to present this information in understandable, clinically useful formats that inform health care practice and improve patient outcomes.''
See also \cite{ana2004role} and \cite{demaria2022role} for further discussion of the role of a medical journal.
}
However, researchers respond to the design of the publication system through the size, length, cost of research studies, the composition of control groups \citep{thorlund2020synthetic}, and in some cases,  which specific findings to report \citep[e.g.,][]{riveros2013timing,shinohara2015protocol}.

We draw a dichotomy between researchers' incentives about (i) whether or how to conduct a study and (ii) which findings to report (e.g., via data manipulation or selective reporting).

We first focus on (i) and abstract from data manipulation and selective reporting (which we defer to Section \ref{sec:p_hacking});
that is, we first suppose that the research design is observable and verifiable as part of the publication process (Section~\ref{sec:design}).
For example, a researcher could choose between experiments with different mean-squared errors and costs, and, constrained by a pre-analysis plan, truthfully report an unbiased estimate of a treatment effect.
The planner can direct attention to the results of any executed study by publishing them;
publication can depend on both a study's design and its results.

We suppose that the publication process affects the supply of research through an individual rationality constraint:
for a researcher to be willing to conduct a study, they must be compensated with a large enough \emph{ex ante} publication probability.
Without these individual rationality constraints, the first-best publication rule would ignore researchers' private costs and publish only results from studies conducted with the lowest possible mean squared error;
due to attention costs,
not all such results would be published.

However, when the research cost of accurate designs is high enough,
the first-best publication rule does not compensate the researcher enough to incentivize them to execute such designs.
As a result, 
the planner faces a trade-off between directing the audience's attention to results that would not deserve it (e.g., treatments with negligible effects)
and rewarding the researcher enough to make them willing to use a costly design. 
The planner may prefer publishing results from less accurate studies---if they are sufficiently less expensive---to avoid having to publish too many results from accurate and costly experiments. 

Returning to our medical example, after providing a new drug to a treatment group, scientists can evaluate its efficacy by using experiments with a different number of participants. A larger experiment is more precise, but more costly for the researcher. 
Allowing smaller experiments can impact on the supply of medical research and new drugs by decreasing research costs.
Our analysis shows that due to the interaction between attention constraints and supply effects, publishing results from smaller experiments can be desirable when experiments are sufficiently costly for researchers to execute.

We next turn in Section~\ref{sec:p_hacking} to settings in which researchers can engage in potentially arbitrary data manipulation or selective reporting.
We suppose that the researcher can chose their research design after learning the (likely) results of their study.
For instance, in the absence of pre-specification,
researchers may engage in forms of $p$-hacking by, e.g., selecting regression specifications or control groups based on the observed outcomes.
We may also be concerned that researchers select experimental environments where they expect treatment effects to be non-representatively large---termed $g$-hacking by \cite{niederle2025experiments}.\footnote{This form of manipulation introduces a form of bias analogous to site selection bias.}

More concretely, we suppose that after observing the study's results, the researcher can report a biased statistic at a (reputational) cost increasing in the bias.
The bias is unobserved by the planner.
The audience is unaware of the possibility of manipulation in publishing findings, so takes published results at face value.
The planner seeks to choose a publication rule to minimize the audience's loss, accounting for the possibility of manipulation.

Each publication rule generates a different degree of manipulation.
For example, suppose that the planner used the publication rule that would be optimal without manipulation---i.e.,
publishing if the reported statistic is above a cutoff.
The researchers would then manipulate their results to reach the cutoff, at a substantial loss for the audience. A different approach that has been proposed is to completely deter manipulation by making publication dependent only on the design, and not on results.\footnote{In practice, this approach can be implemented by committing to publication based on pre-analysis plans, as the \emph{Journal of Development Economics} and the \emph{Journal of Clinical Epidemiology} do (among others).}
However, this approach can incur substantial costs by directing the audience's attention to results that would not substantially affect their actions.

We show that the optimal rule has three key features.
First, it increases the cutoff for which findings always get published compared to settings without manipulation. 
Second, just below this cutoff, it randomizes publication decisions, making the researcher indifferent about whether to (or how much to) manipulate their results.
In particular,
the planner publishes
some findings that would not be published without the possibility of manipulation.
Third, some manipulation does occur in equilibrium for results that would merit attention absent manipulation.
Unlike manipulation under standard cutoff rules, this form of manipulation  benefits the planner by directing attention to results that impact the audience's action enough to merit attention.

To gain some intuition for these features of the optimum,
consider first the optimal publication rule without manipulation.
The planner can eliminate researchers' incentives to manipulate by publishing some results that are below the cutoff, but that would not merit attention. To publish fewer such studies, the planner should increase the cutoff for a result to be guaranteed publication. Increasing the cutoff however also reduces the number of studies that would be published in the absence of manipulation. 
Therefore, the planner then encourages researchers with results that would be published without manipulation, but are below the new cutoff, to engage in some small manipulation to increase their chances of publication.
The loss from publishing some slightly manipulated studies is second-order relative to the gain from publishing more results that substantially impact the audience's~action. 

To formally characterize the optimal publication rule under manipulability,
we formulate the social planner's problem as a mechanism design problem
with ``false moral hazard''
due to the researcher choosing a manipulation after learning the true results. 
The absence of direct transfers and the inability to reward the researcher with a publication probability above one make
the mechanism design problem effectively one with limited transfers.
As a result, standard methods as in \cite{mirrlees1971exploration} and \cite{myerson1981optimal} do not apply.
We solve the mechanism design problem by identifying the precise pattern of binding incentive~constraints.

In Section~\ref{sec:preanalysisplans}, we combine the models from Sections~\ref{sec:design} and~\ref{sec:p_hacking} to ask whether/when the planner should incentivize researchers to run costly experiments that adhere to pre-analysis plans, rather than allow for (cheaper) observational studies.%
\footnote{Our analysis also applies to whether the planner should mandate researchers to send a costly signal (e.g., a detailed pre-analysis plan) that deters them from engaging in manipulation.}
Without accounting for the supply effects of making research more costly, the planner would always require a nonmanipulable experiment \citep{spiess2018optimal,kasy2023optimal}. However,
due to supply effects,
the planner may prefer some observational studies, even when manipulation may occur.

Finally, in Section \ref{Sec:application}, we bring our model to the data and study the optimal publication rules for medical studies, taking into account researchers' best response. We first focus on optimal publication rules in contexts with manipulation. We calibrate the model using about 800,000 $p$-values from studies in hundreds of medical and pharmaceutical journals collected by \cite{head2015extent}. Setting attention costs to capture a standard $5\%$-level $t$-test, we find that the threshold at which results should always be published increases from 1.96  to 2.64. However, many results strictly smaller than 1.96 are published in equilibrium. 
Using the optimal rule makes the average bias of published findings drops by more than a factor of 2. 

We then ask when an observational study with possible manipulation may be preferable to an (equally precise) clinical trial without manipulation, building on the framework of Section~\ref{sec:preanalysisplans}.
Using our calibration, we determine how costly the full clinical trial must be for an observational study to be preferable for the planner.
We compare these costs with estimates of the expected cost-benefit ratio from clinical trials in 13 therapeutic areas. Our calculations illustrate two broad takeaways. First, although experiments dominate observational studies on average in several therapeutic areas, some areas with high research costs may merit the use of observational studies.
However, achieving gains from allowing observational~studies~requires using a different publication standard for them that accounts~for~their~manipulability.

Our results speak to a policy debate regarding the design of controls in clinical trials. 
After providing a new drug to a treatment group, scientists can evaluate its efficacy by using either an experimental placebo group or a synthetic control group obtained from historical medical records \citep{popat2022addressing,yin2022exploring}.
Using a synthetic control group can have large impact on the supply of medical research and new drugs by decreasing research costs \citep[][]{jahanshahi2021use, FDA2023evaluation, wong2014examination}. 
However, using a synthetic control group may increase the estimate's mean-squared error due to lack of randomization, as well as open the door to the manipulation of how the control is specified.
Our analysis highlights that allowing synthetic controls may be desirable if the publication or approval rules is designed to account for their manipulability---a consideration that has not been raised in the current policy debate on the use of synthetic control groups. 

\paragraph*{Related literature.}
This paper connects to a growing literature that develops economic models to analyze statistical protocols. 
In the context of scientific communication, \cite{andrews2019identification},
\cite{abadie2020statistical}, \cite{andrews2021model}, \cite{kitagawa2023optimal}, and (most closely related to our paper) \cite{frankel2022findings} have analyzed how research findings are or should be reported to inform the public.
Our analysis builds on this literature by introducing a model that incorporates researchers' incentives.  This allows us to study how researchers' incentives shape the optimal design of the optimal publication process.

We connect to a broad literature on statistical decision theory \citep[e.g.,][]{wald1950statistical,savage1951theory, manski2004, hirano2009asymptotics, tetenov2012statistical, KitagawaTetenov_EMCA2018} focusing in particular on settings with private researcher incentives.%
\footnote{In addition to the references discussed in detail below, related work in this line includes \cite{chassang2012selective}, \cite{manski2016sufficient}, \cite{banerjee2017decision}, \cite{banerjee2020theory}, \cite{williams2021preregistration}, \cite{bates2022principal,bates2023incentive}, \cite{frankel2022improving}, \cite{libgober2022false}, and \cite{yoder2022designing}. 
}
We develop a unified framework that allows us to analyze both settings in which researchers may choose the research design absent private information, and settings in which researchers can choose the design and manipulate reported findings with private information. This allows us to 
formally study ideas such as when/whether unsurprising results should be published, and whether manipulation should occur in equilibrium, as asked by \cite{glaeser2006researcher}.

In particular, an important distinction from some of these models studying approval decisions, such as \cite{tetenov2016economic}, \cite{bates2022principal, bates2023incentive}, and \cite{viviano2021should}, is that we consider data manipulation and selective reporting.   We also incorporate costs for the researcher of manipulating (and of executing the design), unlike analyses with manipulation by \cite{spiess2018optimal} and \cite{kasy2023optimal}---which we show leads to qualitatively different optimal publication rules.
\cite{mccloskey2022incentive} and \cite{andrews2019identification} propose statistical adjustments for $p$-hacking or publication bias holding researchers' behavior fixed. 
We instead study optimal publication rules in equilibrium, taking into account researchers' best responses.
Last, \cite{henry2019research} and \cite{di2017persuasion,di2021strategic} study decisions with sequential access to the data and with selective sampling, which are different from the question of selective design (and reporting) choice~studied~here.

A large empirical and econometric literature has documented several aspects of the research process, including 
selective reporting, data manipulation, specification search, as well as site selection bias and observational studies' bias 
\citep[e.g.,][]{allcott2015site,olken2015promises, banerjee2020praise,brodeur2020methods, rosenzweig2020external, miguel2021evidence, elliott2022detecting, gechter2022combining, rhys2024much}. Our contribution here is to provide a formal model that studies how incentives interact with several of these choices, shedding light on optimal publication rules and how they compare to ones used in practice.

\section{Setup} \label{sec:setup}

Consider three agents: a researcher, a (representative) audience, and a social planner. The audience and social planner are interested in learning a parameter $\theta \in \mathbb{R}$. 
All agents share a common prior $\theta \sim \mathcal{N}(0, \eta^2)$, whose mean is normalized to 0 without loss of generality.\footnote{Our results continue to hold for $\theta \sim \mathcal{N}(\mu,\eta^2)$, with the audience adjusting their action accordingly.} 

A researcher conducts a study to inform the audience about $\theta$, and seeks to publish their findings. A study is summarized by $(X, \Delta),$ where $\Delta$ denotes the design and $X$ the results observed in the study. If a study is conducted, it will be evaluated according to a publication rule $p(X,\Delta)$ with values in $[0,1]$.
Here, $p(X,\Delta)$ represents the probability of publishing the study, which is assumed to be a Borel measurable function of $(X,\Delta)$.
We assume that 
\begin{equation} \label{eqn:distribution1}
X(\Delta) | \theta \sim \mathcal{N}(\theta + \beta_\Delta, S_{\Delta}^2). 
\end{equation}
Here, $S_\Delta^2$ is the variance of design $\Delta$. 
The quantity $\beta_\Delta$ represents the component of the mean of design $\Delta$ that the audience does not understand;
we refer to $\beta_\Delta$ as the \emph{bias} of design $\Delta$.

Specifically, conditional on publication, the audience forms posterior beliefs about $\theta$ using Bayes' Rule assuming $\beta_\Delta = 0$---i.e., $\theta \sim \mathcal{N}\big(\frac{X \eta^2}{S_\Delta^2 + \eta^2},\frac{S_\Delta^2 \eta^2}{S_\Delta^2 + \eta^2}\big)$.
(We can recenter results to incorporate the component of $\mathbb{E}[X(\Delta)- \theta|\theta]$ that the audience understands.)
Conditional on non-publication, the audience's posterior mean equals its prior mean (zero), as would arise from Bayes' Rule with $p(\cdot)$ symmetric in $X$ and $\beta_\Delta$ symmetric around~0.
Thus, the audience's action $a_p^\star(X,\Delta)$ is 
$$
a_p^\star(X, \Delta) = \begin{cases} 
 \frac{X \eta^2}{S_\Delta^2 + \eta^2} & \text{ if the study is published}  \\ 
0  &\text{ otherwise}
\end{cases}.
$$

Given results $X=X(\Delta)$ for a design $\Delta$, and a parameter $\theta$, the planner incurs a loss 
\begin{equation} \label{eqn:loss}
\begin{aligned} 
\mathcal{L}_p(X, \Delta, \theta) = \mathbb{E}_p\Big[\big(\theta - a_p^\star(X, \Delta)\big)^2\Big] - c_a p(X, \Delta),
\end{aligned} 
\end{equation}
conditional on $X$ and $\theta$,
where $\mathbb{E}_p$ denotes expectation with respect to any stochasticity in the publication decision rule.
This is the (expected) loss of the audience, net of 
attention costs (or shadow costs of attention constraints) proportional to $c_a$.

%

We normalize the value of publication for the researcher to 1.  Given a design $\Delta$, a publication rule $p$, and results $X$, the researcher's expected payoff (conditional on $X$) is
\begin{equation} \label{eqn:payoff_researcher}
v_p(X, \Delta) = p(X, \Delta) - C_\Delta,
\end{equation} 
where $C_\Delta \le 1$ is the researcher's cost of executing design $\Delta$.%
\footnote{We assume $C_\Delta \le 1$ simply to rule out trivial cases in which design $\Delta$ is never chosen by the researcher.}
%
As is standard, whenever the researcher is indifferent between two designs, we implicitly assume she chooses the design that minimizes the planner's expected~loss.

\section{Publication rules under verifiable designs} \label{sec:design}

This section studies optimal publication when the planner can observe the research design, and condition publication on it.
We focus on designs $\Delta$ that are \emph{unbiased} in that $\beta_\Delta = 0$, and defer the analysis of designs that are biased due to manipulation to the following section.


\begin{figure}
 \centering
 \vspace{-24pt}
    \begin{tikzpicture}
\coordinate (1) at (-4,3);
\coordinate (2) at (-2,3);
\coordinate (3) at (-2,5);
\coordinate (4) at (-4,5);
\coordinate (5) at ($(1)!.5!(2)$); 
\coordinate (6) at ($(2)!.5!(3)$);
\coordinate (7) at ($(3)!.5!(4)$);
\coordinate (8) at ($(1)!.5!(4)$);
\coordinate (9) at ($(1)!.5!(3)$);

\coordinate (10) at (-4,0);
\coordinate (11) at (-2,0);
\coordinate (12) at (-2,2);
\coordinate (13) at (-4,2);
\coordinate (14) at ($(10)!.5!(11)$); 
\coordinate (15) at ($(11)!.5!(12)$);
\coordinate (16) at ($(12)!.5!(13)$);
\coordinate (17) at ($(10)!.5!(13)$);
\coordinate (18) at ($(10)!.5!(14)$);

\coordinate (21) at (-10,4);
\coordinate (22) at (3.5,4);
\coordinate (23) at (3.5,5);
\coordinate (24) at (-10, 5);

\coordinate (31) at (-10,3.5);
\coordinate (32) at (-7,3.5);
\coordinate (33) at (-7,2.5);
\coordinate (34) at (-10, 2.5);

\coordinate (41) at (-6.5,3.5);
\coordinate (42) at (-3.5,3.5);
\coordinate (43) at (-3.5,2.5);
\coordinate (44) at (-6.5, 2.5);

\coordinate (51) at (-3,3.5);
\coordinate (52) at (0,3.5);
\coordinate (53) at (0,2.5);
\coordinate (54) at (-3, 2.5);

\coordinate (61) at (0.5,3.5);
\coordinate (62) at (3.5,3.5);
\coordinate (63) at (3.5,2.5);
\coordinate (64) at (0.5, 2.5);

\coordinate (71) at (-10,2);
\coordinate (72) at (-8.6,2);
\coordinate (73) at (-8.6,1);
\coordinate (74) at (-10, 1);

\coordinate (81) at (-8.4,2);
\coordinate (82) at (-7,2);
\coordinate (83) at (-7,1);
\coordinate (84) at (-8.4, 1);

\coordinate (91) at (-6.5,2);
\coordinate (92) at (-5.1,2);
\coordinate (93) at (-5.1,1);
\coordinate (94) at (-6.5, 1);

\coordinate (101) at (-4.9,2);
\coordinate (102) at (-3.5,2);
\coordinate (103) at (-3.5,1);
\coordinate (104) at (-4.9, 1);

\coordinate (111) at (-3,2);
\coordinate (112) at (-1.6,2);
\coordinate (113) at (-1.6,1);
\coordinate (114) at (-3, 1);

\coordinate (121) at (-1.4,2);
\coordinate (122) at (0,2);
\coordinate (123) at (0,1);
\coordinate (124) at (-1.4, 1);

\coordinate (131) at (0.5,2);
\coordinate (132) at (1.9,2);
\coordinate (133) at (1.9,1);
\coordinate (134) at (0.5, 1);

\coordinate (141) at (2.1,2);
\coordinate (142) at (3.5,2);
\coordinate (143) at (3.5,1);
\coordinate (144) at (2.1, 1);

\draw[->] (-8,4.3)  -- (1.8,4.3);

  \node[circle] (g) at (-8,4.3) {$|$};
   \node[circle] (g) at (-5,4.3) {$|$};
     \node[circle] (g) at (-2,4.3) {$|$};
     \node[circle] (g) at (0.5,4.3) {$|$};
  \node[circle] (g) at (-8,3) {Design};
   \node[circle] (g) at (-5,3) {Experiment};
    \node[circle] (g) at (-2,3) {Evaluation};
    \node[circle] (g) at (0.5,3) {Audience};
   \node[circle] (g) at (-8,5) {$\Delta$};

    \node[circle] (g) at (-5,5) {$X \sim \mathcal{N}(\theta, S_\Delta^2)$};

      \node[circle] (g) at (-2,5) {$p(X, \Delta)$};
\node[circle] (g) at (0.5,5) {$a^\star_p(X, \Delta)$};





    \end{tikzpicture}
\vspace{-24pt}
\caption{ Illustration of the variables in the model under observable and verifiable research designs. First, researchers pre-specify the population of interest. Second, they run an experiment and draw a statistic $X$. Third, the planner evaluates the experiment based on a decision rule $p(X, \Delta) \in [0,1]$. Finally, the audience forms a posterior about the estimand of interest, and accordingly takes action $a^\star_p(X,\Delta)$.} \label{fig:time}
\end{figure}


Our analysis proceeds in three steps.
We first characterize the optimal publication rule subject to the constraint of incentivizing the researcher to implement a particular design $\Delta$.
We then characterize which designs are worth incentivizing relative to an outside option.
Last, we characterize the optimal publication rule that chooses between multiple designs.

\subsection{Preliminary analysis for implementing a particular design} \label{sec:simple1}

As a first step, we characterize the constrained optimal publication when the planner must make implementing a particular design $\Delta$ individually rational for the researcher---i.e., when the planner must guarantee that the researcher is willing to conduct the study.

\begin{defn}
\label{defn:1}
A \emph{constrained optimal publication rule} for a design $\Delta$ is a publication rule
$p_\Delta^\star$ that minimizes $\mathbb{E}\left[\mathcal{L}_p(X(\Delta), \Delta, \theta) \right]$ subject to $\mathbb{E}[v_p(X(\Delta), \Delta)] \ge 0.$
\end{defn}

Our first result shows that the constrained optimal publication rule then takes a threshold form, where the threshold $t^\star_\Delta$ for publication depends on the prior, the mean squared error and research cost of the design $\Delta$, and the publication cost.

\begin{prop}[Constrained optimal publication rule] \label{prop:t_test}
If $\Delta$ is an unbiased design,
then a constrained optimal publication rule for $\Delta$ is
the threshold rule $p_\Delta^\star(X) = 1\left\{|X| \ge t^\star_\Delta\right\}$,~where%
\footnote{The proof shows that this rule is uniquely optimal (up to sets of measure 0).}
\[t^\star_\Delta = \min\left\{\frac{S^2_\Delta + \eta^2}{\eta^2} \sqrt{c_a},\left|\Phi^{-1}(C_\Delta/2)\right| \sqrt{S_\Delta^2 + \eta^2}\right\}.\]
\end{prop}

Here, we write $\Phi$ for the cumulative distribution function of a standard normal.

The proof is in Appendix \ref{app:proof1}. 
To understand the intuition behind Proposition~\ref{prop:t_test},
first suppose that the research cost is $C_\Delta = 0$,
so there is no individual rationality constraint for the researcher.
Then, as in \cite{frankel2022findings},
as the planner's publication cost $c_a$ is nonzero,
the planner will publish results that move the audience's optimal action enough to justify incurring the attention costs $c_a$:
i.e., results $|X| \ge \gamma^\star_\Delta,$
where
\[\gamma^\star_\Delta = \frac{S^2_\Delta + \eta^2}{\eta^2} \sqrt{c_a}.\]
When this cutoff rule guarantees an \emph{ex ante} publication probability of at least $C_\Delta$,
the individual rationality constraint does not bind.
In this case, we say the design is cheap.

\begin{defn}
\label{defn:cheap}
An unbiased design $\Delta$ is \emph{cheap} if $C_\Delta < \mathbb{P}\left(|X(\Delta)| \ge \gamma_\Delta^\star\right)$ 
and \emph{expensive} otherwise. 
\end{defn}

Note that higher publication costs $c_a$,
and lower prior variances $\eta^2$,
both raise the threshold $\gamma^\star_\Delta$ and hence make designs more likely to be expensive.
Whether a design is cheap depends on how informative it is (relative to attention costs). 

For expensive designs, 
the cutoff rule from \cite{frankel2022findings} does not provide a large enough \emph{ex ante} publication probability to entice the researcher to conduct the study in the first place.
Hence, 
the planner needs to commit to publishing more results in order to satisfy the researcher's individual rationality constraint.
It is optimal for the planner to publish results that move the audience's action the most, even if these results do not move the audience's action enough to justify the attention cost $c_a$.
Hence, the planner sets a cutoff that ensures an \emph{ex ante} publication chance of $C_\Delta$---i.e., a cutoff of
\[\left|\Phi^{-1}(C_\Delta/2)\right| \sqrt{S_\Delta^2 + \eta^2}.\]
This second cutoff is below $\gamma^\star_\Delta$ for (and only for) expensive designs, and is the optimal cutoff for such designs. 
Thus, if research costs are large enough that the implemented design is expensive,
the researcher's incentives play a central role in determining the optimal publication rule, unlike in \cite{frankel2022findings}.

The following corollary summarizes and formalizes the preceding discussion.  

\begin{cor}
\label{cor:t_testCheapExpensive}
\begin{enumerate}[label=(\alph*)]
    \item If $\Delta$ is a cheap design, then the cutoff for a constrained optimal publication rule is $t^\star_\Delta = \gamma^\star_\Delta$.
    \item If $\Delta$ is an expensive design, then the cutoff for a constrained optimal publication rule is $t^\star_\Delta = \left|\Phi^{-1}(C_\Delta/2)\right| \sqrt{S_\Delta^2 + \eta^2}$. 
\end{enumerate}
\end{cor}

\subsection{Which designs are ever worth incentivizing}

As a second step, we study when a design is worth incentivizing relative to the outside option of no study.
If it is not worth doing so, then it is not worth making it individually rational for the researcher to conduct research based on design $\Delta$ in optimum.

Let $\mathcal{L}^\star_\Delta = \mathbb{E}\left[\mathcal{L}_{p_\Delta^\star}(X(\Delta), \Delta, \theta) \right]$
denote the optimal expected loss for the planner once implementing design $\Delta$.
(Here $p_\Delta^\star$ is a constrained optimal publication rule for $\Delta.$)
The expected loss if no research is published is the prior variance $\eta^2$.
Comparing these two quantities determines whether a design is worth incentivizing in the first place.

\begin{defn}
A design $\Delta$ is \emph{worthwhile} if $\mathcal{L}^\star_\Delta \le \eta^2$.
\end{defn}

We next characterize which designs are worthwhile.
If a design is cheap,
then the planner can selectively publish only results that move the audience's beliefs enough to justify incurring the attention cost.
Thus,
the \emph{ex post} loss under the constrained optimal publication rule is always lower than $\eta^2$, and so the (\emph{ex ante}) expected loss is less than $\eta^2$.

\begin{prop}[When are cheap designs worthwhile?]
\label{prop:cheapWorthwhile}
Every cheap design $\Delta$ is worthwhile.
\end{prop}

The proof is in Appendix \ref{proof:prop:cheapWorthwhile}. 
Whereas cheap design are always worthwhile, 
for expensive designs,
the situation is more delicate.
Incentivizing the researcher to implement a design requires committing to publish results that the planner would \emph{ex post} prefer not to publish.
When attention costs $c_a$ are large enough,
the cost of publishing these marginal results outweigh the benefit of publishing results that substantially affect the audience's action.
How large $c_a$ needs to be for this to occur depends on the design's cost and variance.

To formalize this intuition, it will be convenient to express our results in terms of the difference between the posterior and prior variances conditional on publication of the results of a design $\Delta$, which we denote~by
\newcommand\postvarred[1]{\operatorname{PostVarRed}(#1)}
\[\postvarred{\Delta} := \eta^2 - \frac{S_\Delta^2\eta^2}{S_\Delta^2 + \eta^2} = \frac{\eta^4}{S_\Delta^2 + \eta^2}.\]
This quantity is a measure of the informativeness of a design:
it represents how much learning the results of the design improves the expected utility of a Bayesian audience with $L^2$ loss.
Note that $\postvarred{\Delta}$ is increasing in $\eta^2$ and decreasing in  $S_\Delta^2$.
Whether a design is worthwhile then depends on how $\postvarred{\Delta}$ compares to the product of the attention and research costs, up to a small remainder that vanishes with large research costs ($C_\Delta \approx 1$).\footnote{Proposition \ref{cor:worthwhileNeccSuff} is sharp up to the term $\eta^2 (1 - C_\Delta)^3$; Lemma~\ref{lem:worthwhileNeccSuffUgly} in Appendix \ref{proof:cor:worthwhileNeccSuff} provides a more involved condition that is both necessary and sufficient for worthwhileness.}

\begin{prop}[When are expensive designs worthwhile?]
\label{cor:worthwhileNeccSuff}
Let $\Delta$ be an expensive design.
\begin{enumerate}[label=(\alph*)]
\item If $\postvarred{\Delta} \ge C_\Delta c_a + \eta^2 (1 - C_\Delta)^3$,
then $\Delta$ is worthwhile. 
\item If 
$\postvarred{\Delta} < C_\Delta c_a$,
then $\Delta$ is not worthwhile. 
\end{enumerate}
\end{prop}

The proof is in Appendix \ref{proof:cor:worthwhileNeccSuff}.  
In particular, Proposition~\ref{cor:worthwhileNeccSuff} shows that
as $C_\Delta$ or $c_a$ increases, the posterior variance reduction must increase proportionally for a design to remain worthwhile (up to a small remainder that vanishes in the case of large research costs).

\subsection{Choosing which design to incentivize} \label{sec:simple2}

We next study the optimal publication rule when there is more than one possible design.
Without loss of generality, we suppose that both designs are worthwhile, and that
the one with a lower mean squared error has a higher research cost,
so the planner faces a non-trivial problem about which design to incentivize.

\begin{setting} \label{set:1} 
Researchers can choose between two 
designs $E,O$ that are unbiased and worthwhile.
The designs have mean squared errors $S_E^2 < S_O^2$ and costs $C_E > C_O$.
\end{setting} 

We think of $\Delta = E$ as a possibly expensive experiment and $\Delta = O$ as a lower-cost experiment or non-manipulable observational study (manipulation is studied in Section~\ref{sec:p_hacking}).

\begin{exmp}[Low-cost and costly experiment] 
Suppose that $O$ corresponds to an experiment with fewer participants than $E$.
In this case, we have $C_O < C_E$ and $S_O^2 > S_E^2$. \qed 
\end{exmp}

\begin{exmp}
[Experiment versus nonmanipulable observational study]
\label{exmp:experiment}
Suppose that $O$ corresponds to an observational study with no strategic manipulation of the results. 
Let $X(E) = \theta + \varepsilon_E$ and $X(O) = \theta + b_O + \varepsilon_O$, where $\varepsilon_E \sim \mathcal{N}(0, S_E^2)$ denote the estimation noise from the experiment, $\varepsilon_O \sim \mathcal{N}(0,\sigma_O^2)$ denotes the idiosyncratic noise from the experiment or observational study in $O$, and $b_O | \varepsilon_O \sim \mathcal{N}(0,\sigma_B^2)$ denotes a random effect, which captures unobserved bias drawn from a fixed (Gaussian) distribution.%
\footnote{For example, \cite{rhys2024much} investigates the distribution of $b_O$ through a meta-analysis. 
} 
We then have that $X(O) \sim \mathcal{N}(\theta,S_O^2)$, where the mean-squared error $S_O^2 = \sigma_O^2 + \sigma_B^2$
includes both sampling uncertainty $\sigma_O^2$ and irreducible error $\sigma_B^2$ arising from the variance of the bias.
\qed 
\end{exmp}

We next use Proposition~\ref{prop:t_test} to study  the optimal choice between the two designs.
Because the design $\Delta$ is observable by the planner and verifiable as part of the publication process, the planner can incentivize their preferred design by setting
\begin{equation} \label{eqn:optimal2}
p^\star(X, \Delta) = 1\{\Delta = \Delta^{\mathrm{planner}}\} p_\Delta^\star(X) \qquad \text{where} \qquad \Delta^{\mathrm{planner}} \in \argmin_{\Delta \in \{E, O\}} \mathcal{L}^\star_\Delta. 
\end{equation}
For instance, the planner may only accept experiments with a minimum level of precision. It is immediate that $p^\star(X, \Delta)$ minimizes the planner's expected loss. 
%
We therefore study the optimal design choice by comparing the minimized loss of the social planner when implementing the experiment versus implementing the observational study. 
More generally,
we can use similar logic to compare the effectiveness of any two designs.

\begin{defn}
Design $\Delta$ is \emph{planner-preferred} to design $\Delta'$ if $\mathcal{L}^\star_{\Delta} < \mathcal{L}^\star_{\Delta'}$.
\end{defn}

It is immediate that if a design is planner-preferred to a worthwhile design $\Delta'$, then $\Delta$ is worthwhile.
In particular, Proposition~\ref{prop:cheapWorthwhile} implies that if a design $\Delta$ is planner-preferred to a cheap, unbiased design $\Delta'$,
then $\Delta$ is worthwhile.

If the more precise experiment $E$ is cheap, then its higher research cost is irrelevant to the planner.
Therefore, the experiment is planner-preferred to $O$.

\begin{prop}
\label{prop:cheapPreferred}
In Setting~\ref{set:1},
if $E$ is cheap, then $E$ is planner-preferred to $O$.
\end{prop}

The proof is in Appendix \ref{proof:prop:cheapPreferred}. 
This result implies that it suffices to compare the mean-squared error of two cheap studies to identify which one is planner-preferred.  

When the experiment $E$ is expensive, the situation is more delicate.
Implementing $E$ requires committing to publish more results, which may be costly for a planner.
When the publication or attention costs $c_a$ are large enough,
the costs of publishing more results outweighs the benefits of a more precise design.

\begin{prop}
\label{prop:cpPreferred}
In Setting~\ref{set:1}, if $E$ and $O$ are expensive, then
there exists a threshold $c_a^\star(E,O,\eta) > 0 $ such that $E$ is planner-preferred to $O$ if and only if $c_a < c_a^\star(E,O,\eta)$,~where 
\[c_a^\star(E,O,\eta) = \frac{\postvarred{E} - \postvarred{O}}{C_E - C_O} - \eta^2 \frac{O\big((1-C_E)^3\big) - O\big((1-C_O)^3\big)}{C_E - C_O}.\] 
\end{prop} 

The proof is in Appendix \ref{proof:prop:cpPreferred}. 
Proposition \ref{prop:cpPreferred} shows that to choose between $E$ and $O$, assuming $E$ and $O$ have high research costs ($(1 - C_O)^3 \approx 0$), it suffices to compare
\begin{equation} \label{eqn:approximation_c}
\postvarred{E} - \postvarred{O} \text{ versus }(C_E - C_O) c_a
\end{equation}
up to a small remainder.
That is, we must compare the difference in the posterior \pagebreak variance reductions to the difference in research costs, adjusted by the attention cost $c_a$. A larger $c_a$ favors less costly designs. 
The following theorem formalizes these intuitions
and sharpens them to apply even for smaller research costs.

\begin{theorem}[Comparing two designs] \label{prop:main_comparison}
In Setting~\ref{set:1}, suppose that $E$ and $O$ are expensive.
\begin{enumerate}[label=(\alph*)]
\item If $\postvarred{E} - \postvarred{O} \ge \big(1 - \frac{C_O}{C_E}\big)c_a,$
then $E$ is planner-preferred to $O$.
\item If $\postvarred{E} - \postvarred{O} \le \big(C_E - \frac{1 + 2C_O}{3} \big) c_a$, then $O$ is planner-preferred~to~$E$.
\end{enumerate}
If instead $O$ is cheap, then (a) and (b) hold with $C_O$ replaced by $P(|X(O)| \ge \gamma_O^\star)$. 
\end{theorem}

The proof is in Appendix \ref{proof:prop:main_comparison}. 
Intuitively the comparison between two studies must depend on the posterior variance reduction of each study (which itself depends their mean-squared error) and the costs of each study.

\begin{figure}[b!]
    \centering
    \includegraphics[scale=0.5]{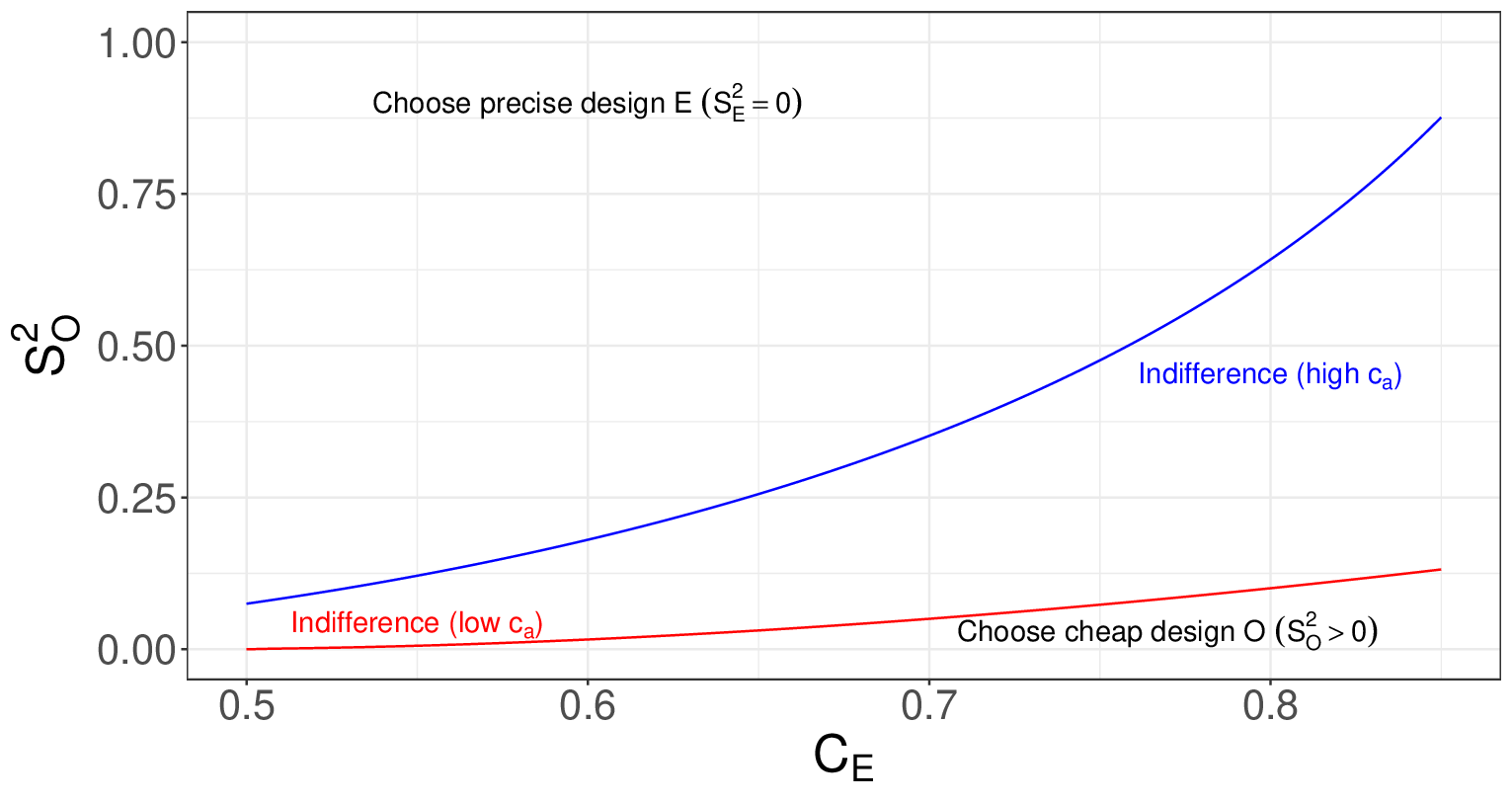}
    \caption{Comparison of two designs.  The figure depicts which of two designs is planner-preferred in Setting~\ref{set:1} for ranges of variances $S_O^2$ for the less precise design and research costs $C_E$ for the more precise design.  The less precise design is assumed to be cheap, the more precise design is assumed to have zero variance ($S_E = 0$), and we normalize $\eta^2 = 1$.  Planner preference is computed using exact expression for $\mathcal{L}_\Delta^\star$ from Lemma~\ref{lem:taylor_large} in Appendix~\ref{app:preliminary}.
    The more precise design $E$ is more likely to be preferred for higher $S_O^2$ and lower $C_E$.
    The red (resp.~blue) curve depicts design parameters for which the planner is indifferent between the two designs for a low attention cost $c_a = 0.5$ (resp.~high attention cost $c_a = 1$). 
    } 
    \label{fig:fig2}
\end{figure}

As the cost of attention $c_a$ increases, the planner's preference shifts from a more accurate design 
to a less accurate design with a smaller cost. This is because, for costly studies, the planner must internalize not only the effect of the mean-squared error on the audience's loss function, but also the research cost associated with the study. With high attention costs (large $c_a$), more costly experiments impose more stringent constraints on the publication rules, making those undesirable for the planner. For example, in this case, medical studies with smaller experiment size may be preferred over more precise experiments when variable costs are sufficiently large. 

Theorem \ref{prop:main_comparison} provides bounds that are not tight to enhance interpretability. However, more involved necessary and sufficient conditions for each design to be planner-preferred can be obtained directly from a result in the Appendix (Lemma~\ref{lem:taylor_large} in Appendix~\ref{app:preliminary}). Using these conditions, 
Figure  \ref{fig:fig2} reports the indifference curves between two experiments with different mean-squared errors and costs. Figure  \ref{fig:fig2}  shows how the planner's preference shifts towards cheaper (and noisier) designs even when the experiment has zero variance.

\begin{rem}[Calibrating parameters] As we illustrate in Section~\ref{Sec:application}, we can calibrate $\eta^2$ as the prior variance of parameters obtained from meta-studies \citep[see, e.g.,][]{rhys2024much,bartovs2023empirical}. The cost $c_a$ can be calibrated to make the cutoff $\gamma^\star_E$ for cheap experiments match the critical value of a $t$-test for a particular level---for example, $\sqrt{c_a} \frac{\eta^2 + 1}{\eta^2} = 1.96$. \qed 
\end{rem} 

\begin{rem}[Internalizing research costs in the planner's objective] A variant of the model considered here is to let the planner internalize the research costs. For instance, one could consider a planner's objective $\mathcal{L}_\Delta^\star + C_\Delta$. Our main analysis continues to hold in this setting, by taking into account that the comparisons in Equation \eqref{eqn:approximation_c} and Theorem \ref{prop:main_comparison} must account for the additional cost component in the planner's loss function. \qed 
\end{rem}

\section{Publication rules under non-verifiable designs} 
\label{sec:p_hacking}

In this section, we turn to settings where researchers may engage in data manipulation and selective reporting.  To this end, we investigate optimal publication rules when researchers choose the research design $\Delta$ using knowledge of the statistics drawn in the experiment. Here, the design $\Delta$, and its corresponding bias $\beta_\Delta$, are known to the researcher but not verifiable by the social planner.
Specifically, we consider the following setting.

\begin{setting} \label{set:model_p_hacking}
The class of designs is $\Delta \in \mathbb{R}$, with $S_\Delta^2 = S^2$ common knowledge, and $\beta_\Delta = \Delta$ known to the researcher.  Writing $X(\Delta) = \theta + \beta_\Delta + \varepsilon$ (where $\varepsilon | \theta \sim \mathcal{N}(0, S^2)$),
the researcher observes $\theta + \varepsilon$ and chooses $\Delta$ to maximizes her realized payoff $v_p(X(\Delta), \Delta)$. Research costs are given by  $C_\Delta = c_m |\beta_\Delta| + C_0$, where $0 < c_m < \infty$ and $C_0 < 1$.  The social planner chooses a (Borel measurable) publication rule $p(X,\Delta) = p(X)$ as a function of $X$ only. 
\end{setting} 

Figure \ref{fig:time2} illustrates the model: the researcher deterministically chooses the bias of the reported statistic. They, however, pay a cost $C_\Delta$ increasing in the bias.
The component $c_m|\beta_\Delta|$ of the cost $C_\Delta$ captures reputational or computational costs associated with the manipulation, assumed to be increasing and linear in the magnitude $|\beta_\Delta|$ of the bias. The component $C_0$ captures a fixed cost.\footnote{It is possible to also incorporate fixed costs of manipulation (e.g., $C(\Delta) = c_m|\beta_\Delta| + c_f 1\{\beta_{\Delta} \not= 0\} + C_0$) to capture costs of introducing any manipulation, or to include nonlinear costs of manipulation, though the specific characterization of the optimal publication rule would then be different.}
The researcher observes $\theta + \varepsilon$, and hence maximizes realized utility conditional on the observed statistics when choosing~$\Delta$.

We think of the researcher's action of deterministically choosing the bias as a stylized description of data manipulation or selective reporting, whereby researchers can choose their research design after learning the results of potential studies $X$. In practice, in the absence of a precise pre-specification, researchers can change the covariates in a regression, winsorize the data in particular ways, or make other design choices functions of the statistics.
These manipulations are all forms of $p$-hacking, and bias results in a way that is difficult or impossible for the planner to verify.%
\footnote{%
For precise experiments (i.e., when $X \approx \theta)$, our model also speaks to a related form of manipulation, which \cite{niederle2025experiments} terms $g$-hacking.   In particular, researchers may choose the specific experimental environment to be one where they expect treatment effects to be non-representatively large using private information about $\theta$---which they may obtain from piloting, theory, or intuition.
For example, researchers can choose the experimental site for a field experiment, or an experimental setting in a lab experiment, in which they expect treatment effects to be large.
These manipulations introduce bias in results analogous to site selection bias,
and are also difficult or impossible for the planner to verify.
}

\begin{figure}[t!]
 \centering
    \begin{tikzpicture}

\coordinate (1) at (-4,3);
\coordinate (2) at (-2,3);
\coordinate (3) at (-2,5);
\coordinate (4) at (-4,5);
\coordinate (5) at ($(1)!.5!(2)$); 
\coordinate (6) at ($(2)!.5!(3)$);
\coordinate (7) at ($(3)!.5!(4)$);
\coordinate (8) at ($(1)!.5!(4)$);
\coordinate (9) at ($(1)!.5!(3)$);

\coordinate (10) at (-4,0);
\coordinate (11) at (-2,0);
\coordinate (12) at (-2,2);
\coordinate (13) at (-4,2);
\coordinate (14) at ($(10)!.5!(11)$); 
\coordinate (15) at ($(11)!.5!(12)$);
\coordinate (16) at ($(12)!.5!(13)$);
\coordinate (17) at ($(10)!.5!(13)$);
\coordinate (18) at ($(10)!.5!(14)$);

\coordinate (21) at (-10,4);
\coordinate (22) at (3.5,4);
\coordinate (23) at (3.5,5);
\coordinate (24) at (-10, 5);

\coordinate (31) at (-10,3.5);
\coordinate (32) at (-7,3.5);
\coordinate (33) at (-7,2.5);
\coordinate (34) at (-10, 2.5);

\coordinate (41) at (-6.5,3.5);
\coordinate (42) at (-3.5,3.5);
\coordinate (43) at (-3.5,2.5);
\coordinate (44) at (-6.5, 2.5);

\coordinate (51) at (-3,3.5);
\coordinate (52) at (0,3.5);
\coordinate (53) at (0,2.5);
\coordinate (54) at (-3, 2.5);

\coordinate (61) at (0.5,3.5);
\coordinate (62) at (3.5,3.5);
\coordinate (63) at (3.5,2.5);
\coordinate (64) at (0.5, 2.5);

\coordinate (71) at (-10,2);
\coordinate (72) at (-8.6,2);
\coordinate (73) at (-8.6,1);
\coordinate (74) at (-10, 1);

\coordinate (81) at (-8.4,2);
\coordinate (82) at (-7,2);
\coordinate (83) at (-7,1);
\coordinate (84) at (-8.4, 1);

\coordinate (91) at (-6.5,2);
\coordinate (92) at (-5.1,2);
\coordinate (93) at (-5.1,1);
\coordinate (94) at (-6.5, 1);

\coordinate (101) at (-4.9,2);
\coordinate (102) at (-3.5,2);
\coordinate (103) at (-3.5,1);
\coordinate (104) at (-4.9, 1);

\coordinate (111) at (-3,2);
\coordinate (112) at (-1.6,2);
\coordinate (113) at (-1.6,1);
\coordinate (114) at (-3, 1);

\coordinate (121) at (-1.4,2);
\coordinate (122) at (0,2);
\coordinate (123) at (0,1);
\coordinate (124) at (-1.4, 1);

\coordinate (131) at (0.5,2);
\coordinate (132) at (1.9,2);
\coordinate (133) at (1.9,1);
\coordinate (134) at (0.5, 1);

\coordinate (141) at (2.1,2);
\coordinate (142) at (3.5,2);
\coordinate (143) at (3.5,1);
\coordinate (144) at (2.1, 1);

\draw[->] (-8,4.3)  -- (1.8,4.3);

  \node[circle] (g) at (-8,4.3) {$|$};
   \node[circle] (g) at (-5,4.3) {$|$};
     \node[circle] (g) at (-2,4.3) {$|$};
     \node[circle] (g) at (0.5,4.3) {$|$};
  \node[circle] (g) at (-8,3) {Statistic};
   \node[circle] (g) at (-5,3) {Manipulation};
    \node[circle] (g) at (-2,3) {Evaluation};
    \node[circle] (g) at (0.5,3) {Audience};
   \node[circle] (g) at (-8,5) {$X' = \theta + \varepsilon$};

    \node[circle] (g) at (-5,5) {$X = X' + \beta_\Delta$};

      \node[circle] (g) at (-2,5) {$p(X)$};
\node[circle] (g) at (0.5,5) {$a^\star(X)$};






    \end{tikzpicture}
\vspace{-24pt}
\caption{Illustration of the variables in the model with unverifiable research designs and researcher private information. First, researchers observe the vector of statistics. They then manipulate the design by introducing a bias into the statistics and maximize their private utility. The social planner does not observe the bias, and evaluates the study based on a publication rule $p(X)$ that only depends on the statistics $X$.} \label{fig:time2}
\end{figure}

As we discuss in Section \ref{sec:setup}, the audience updates their beliefs assuming that $\beta_\Delta = 0$.
Thus, we assume that the audience is unaware of the possibility of data manipulation---i.e., that they take published findings at face value.
For example, in our medical example, the audience may represent doctors or policymakers who are not familiar with (or would need to incur high costs to understand) experimental details.
By contrast, the planner is aware of the possibility of manipulation, and that the audience is unaware of it, and minimizes the audience's loss taking both of these points into account.

We assume that the variance of the residual noise $\varepsilon$ equals $S^2$ for all designs $\Delta$. 
We interpret this assumption as stating that standard errors are verifiable as part of the publication process; hence, we focus on manipulation that introduces unverifiable bias in reported~results.

\subsection{Optimal publication rule under manipulation}

The planner knows $S^2$, cannot observe or verify $\beta_\Delta$, and minimizes expected loss over $(\theta,\varepsilon)$ taking into account the researcher's (endogeneous) incentives to manipulate their results.
Formally, writing
$\mathcal{P}$ for the set of all Borel measurable functions $p(X,\Delta)$ that are constant in $\Delta$ (i.e., do not depend on the design),
an optimal publication rule is defined by 
\begin{equation} \label{eqn:p_star_manipulation}
p^\star \in \argmax_{p \in \mathcal{P}} \mathbb{E}_{\theta,\varepsilon}\left[\mathcal{L}_p(X(\Delta_p^{\star}), \Delta_p^{\star}, \theta)\right] \quad \text{with} \quad \Delta_p^\star \in \argmax_{\Delta} v_p(X(\Delta),\Delta).
\end{equation} 
Here, $\Delta_p^{\star}$ denotes an optimal response of the researcher to the publication rule given $\theta + \varepsilon$. 

The main result of this section shows that the optimal publication rule is in a class of smoothed cutoff rules.
Intuitively, a linearly smoothed cutoff rule is a deterministic publication rule below and above thresholds $X^\star - \frac{1}{k}$ and $X^\star$, respectively; it randomizes the publication chances between these two thresholds, with publication probability increasing linearly the value of the reported statistic $|X|$ with slope $k$.

\begin{defn}
A \emph{linearly smoothed cutoff rule} with cutoff $X^\star$ and slope $k$ is defined by
\[p_{X^\star,k}(X) = \begin{cases}
0 & \text{if } |X| \le X^\star - \frac{1}{k}\\
1-k(X^\star - |X|) & \text{if } X^\star - \frac{1}{k} < |X| < X^\star\\
1 & \text{if } |X| \ge X^\star
\end{cases}.\]
\end{defn}

The special case of slope $k = \infty$ and threshold $X^\star = \gamma^\star = \frac{S^2 + \eta^2}{\eta^2} \sqrt{c_a}$ corresponds to a publication rule for cheap experiments without manipulation (Corollary~\ref{cor:t_testCheapExpensive}).

We next characterize the optimal publication rule in settings with manipulation.

\begin{theorem}[Optimal publication rule under unverifiable designs] \label{thm:optimal_assymetric_info} 
In Setting~\ref{set:model_p_hacking}:
\begin{enumerate}[label=(\alph*)]
\item There exists a cutoff $X^\star \in \left(\gamma^\star,\gamma^\star + \frac{1 - C_0}{c_m}\right)$ such that
the linearly smoothed cutoff rule
$p_{X^\star,c_m}$ is optimal.
\item For each optimal publication rule $p$, there exists $X^\star \in \left(\gamma^\star,\gamma^\star + \frac{1 - C_0}{c_m}\right)$ such that $p(X) = p_{X^\star,c_m}(X)$ (resp.~$p(X) \le C_0$) for almost all $X \ge 0$ with $p_{X^\star,c_m}(X) > C_0$  (resp.~$p_{X^\star,c_m}(X) \le C_0$).
\end{enumerate}
\end{theorem}


The proof is in Appendix \ref{proof:thm:1}. As publication probabilities are between 0 and 1, the mechanism design problem is effectively one with limited transfers.
As a result, standard techniques to eliminate transfers from the planning problem \cite{mirrlees1971exploration} and \cite{myerson1981optimal} do not apply,
and we need to deal directly with both publication probabilities and equilibrium manipulation.
To solve the mechanism design problem, we identify the precise pattern of binding incentive constraints, which are upward incentive constraints to and from type $\gamma^\star$.
Constraining the equilibrium utility level for that type to be $u$, we show that the optimal publication rule is a linearly smooth cutoff rule.
We then optimize over $u$ to bound the optimal cutoff $X^\star$.
Although the optimal cutoff $X^\star$ does not admit a simple closed-form expression, it can be computed numerically---as we show in Figure \ref{fig:manipulation}.

\subsection{Interpretation and implications for published findings}

To provide intuition for the structure of the optimal publication rule in Theorem \ref{thm:optimal_assymetric_info}, we illustrate how the possibility of manipulation affects the optimal publication rule.
For ease of exposition, we abstract from fixed research costs in our discussion (i.e., take $C_0 = 0$).
Our formal results all apply to the case of general $C_0$,
and all proofs are in Appendix~\ref{proof:prop:non_surprising}.

Suppose first that the social planner ignored the possibility of manipulation, and set a cutoff rule for a cheap experiment as in Corollary~\ref{cor:t_testCheapExpensive}. Then we would observe bunching around the publication cutoff $\gamma^\star$, as researchers with $|\theta + \varepsilon| \in \big(\gamma^\star - \frac{1}{c_m}, \gamma^\star\big)$ would introduce a bias to publish. Researchers with $|\theta + \varepsilon| < \gamma^\star - \frac{1}{c_m}$ would find it unprofitable to introduce any bias (as the cost would not compensate the benefits) and therefore would not publish.
The first line of Table~\ref{tab:sequence} and the first two panels of Figure~\ref{fig:manipulation} summarize this discussion.

\begin{table}[p!] \centering
\centering 
\begin{tabular}{@{\extracolsep{1pt}} ccccccc} 

Publication rule & Testable observation & Published results & Manipulation
\\ \hline 
\hline \\[-1.8ex] 
Optimal cutoff rule  & Large bunching  & Only results  & Large \\ ignoring manipulation & & with $|X| \ge \gamma^\star$ & \\ \hline
Add randomization  & No bunching & Many results  & None \\ below cutoff & & with $|X| < \gamma^\star$ &  \\ \hline
Optimal rule & Some bunching  & Some results  & Some \\ (Randomize + raise cutoff) & & with $|X| < \gamma^\star$  &
\end{tabular}
\caption[Caption for LOF]{Comparisons between three different publication rules. The first row corresponds to a publication rule that chooses the optimal cutoff $\gamma^\star$ assuming no manipulation. In this case, we observe large bunching at the cutoff. The second row corresponds to a publication rule that removes manipulation by introducing randomization below the cutoff $\gamma^\star$. This rule is suboptimal as too many results that do not merit attention get published. The last row corresponds to the optimal publication rule, which randomizes using a higher cutoff $X^\star > \gamma^\star$.
}
  \label{tab:sequence} 
\end{table}

\begin{figure}[p!]
\centering
\hspace{-21pt}
\includegraphics[scale = 0.6]{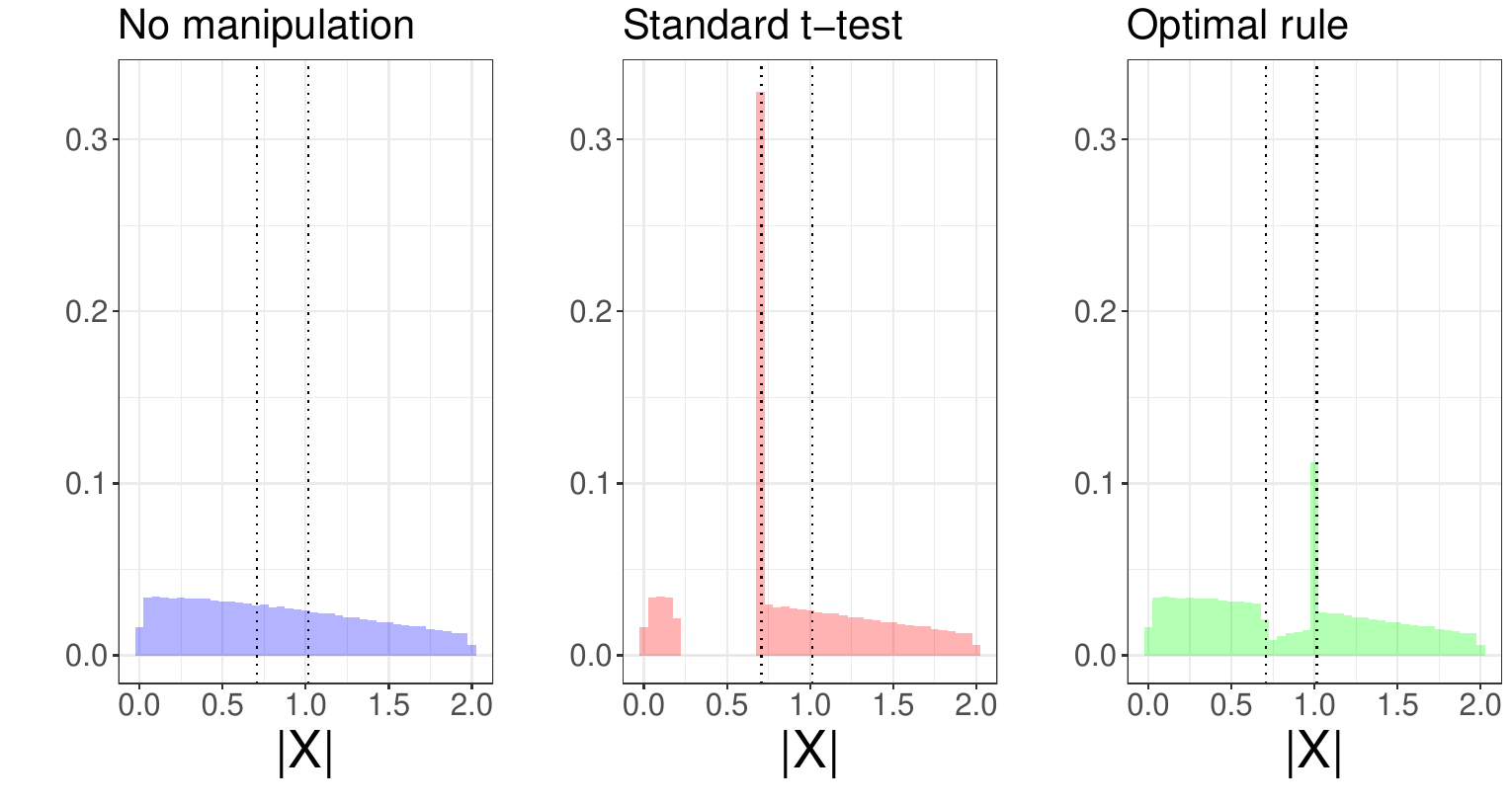}
\caption{
Distribution of $|X|$ under different publication regimes. In contexts where no manipulation is possible (left panel), we see no bunching on the distribution of $|X|$. When instead we consider a standard $t$-test rule that it optimal without manipulation, we observe  a large bunching around the critical cutoff $\gamma_E^\star$ (center panel). With the optimal publication rule under manipulation (right panel), the critical threshold for which findings are published with probability 1 is higher, whereas results in the neighborhood of $X^\star$ are published with positive probability. We observe some bunching at $X^\star$, though less extreme than the bunching at $\gamma^\star$ under the standard $t$-test rule. 
We use parameter values $c_m = 2$, $\eta^2 = 2$, $S^2 = 0$ and $c_a = 0.5$. The first dotted line corresponds to $t$-test cutoff $\gamma^\star$, and the second dotted line corresponds to the cutoff $X^\star$ for the optimal rule.} \label{fig:manipulation}
\end{figure}

Next, suppose that the planner introduces randomization in the publication rule whenever $|X| \in \big(\gamma^\star - \frac{1}{c_m}, \gamma^\star\big)$ as in a linearly smoothed cutoff rule with cutoff $\gamma^\star$. This randomization makes the researcher indifferent between manipulating and not manipulating the data, at the cost of publishing some results below $\gamma^\star$, which do not move the audience's action enough to justify incurring the attention cost.
The second line of Table~\ref{tab:sequence} summarizes this discussion.

More generally, the optimal publication rule randomizes publication for some unmanipulated results below $\gamma^\star$. This is in stark contrast with the case without manipulation.  

\begin{prop}[Some results that do not merit attention are published despite not being manipulated]
\label{prop:non_surprisingPublish}
In Setting~\ref{set:model_p_hacking},
consider any optimal publication rule $p^\star$.
For some types $|\theta + \varepsilon| < \gamma^\star$, we have $p^\star(X(\Delta^\star_{p^\star})) > C_0$ but $\beta_{\Delta^\star_{p^\star}} = 0$.
\end{prop}

However, simply randomizing for results below $\gamma^\star$ is still suboptimal as too many unsurprising results are published in the randomization regime.
The last step is to increase the threshold $X^\star$ to lower the loss from publishing results that do not move the audience's action enough.
A consequence is that some results that merit attention are not published.


\begin{prop}[Some results that merit attention are not published]
\label{prop:surprisingNotPublish}
In Setting~\ref{set:model_p_hacking}, consider any optimal publication rule $p^\star$.
For some types $\theta + \varepsilon > \gamma^\star$, we have $p^\star(X(\Delta^\star_{p^\star})) < 1$.
\end{prop}

Given that results are only guaranteed publication if they cross a higher threshold than $\gamma^\star$, some manipulation can be beneficial to the planner to increase the publication rate of surprising findings. Therefore, in the planner's preferred equilibrium, researcher types below the cutoff $X^\star$ and above $\gamma^\star$ engage in some manipulation.  This form of manipulation is distinct from the manipulation that researchers engage in under (non-smoothed) cutoff rules, which involves results that should not be published and therefore hurts the planner.

\begin{prop}[Manipulation in equilibrium] \label{prop:manipulation_equilibrium}
In Setting~\ref{set:model_p_hacking},
consider any optimal publication rule, and let $X^\star$ be as in Theorem~\ref{thm:optimal_assymetric_info}(b).
For almost all $\theta + \varepsilon \in (\gamma^\star, X^\star)$, we have $\beta_{\Delta_{p^\star}^\star} > 0$.
\end{prop}

There is no point in manipulating beyond the cutoff $X^\star$, as results $X^\star$ are published with probability 1.
As a result, there is bunching at the cutoff $X^\star$ in optimum.

\begin{prop}[Bunching at $X^\star$ in equilibrium] \label{prop:bunching_equilibrium}
In Setting~\ref{set:model_p_hacking},
consider any optimal publication rule, and let $X^\star$ be as in Theorem~\ref{thm:optimal_assymetric_info}(b).
There exists $\zeta > 0$ such that $\theta + \varepsilon + \beta_{\Delta_{p^\star}^\star} = X^\star$ for almost all  $\theta + \varepsilon \in (X^\star - \zeta,X^\star)$.
\end{prop}

The third line of Table~\ref{tab:sequence} and the third panel of Figure~\ref{fig:manipulation} summarize Propositions~\ref{prop:manipulation_equilibrium} and~\ref{prop:bunching_equilibrium}.

\begin{rem} An unrelated application of our analysis is to the structure of the lottery for top picks in the National Basketball Association (NBA) draft. To help weak teams become more competitive, the NBA seeks to offer a higher chance of top picks to the lowest-ranked teams. This creates incentives for the teams near the bottom to manipulate their quality ($\theta$) by losing matches so to be classified as the worst-team (i.e., lowering $X$). To disincentive this form of manipulation, in 2019, the NBA began to offer the same chances at top picks to the three worst-ranked teams.  This form of randomization has a similar structure to the randomization features in our optimal publication rule under manipulation.  
\end{rem}

\section{Taking stock: Research costs versus manipulability}
\label{sec:preanalysisplans}

In this section, we combine the analyses of the previous two sections to study when the planner may want to incentivize a manipulable observational study instead of a more costly, but pre-registered,  experiment.
More precisely, we consider the following setting.

\begin{setting}[Pre-specification versus possible manipulation] Consider the following two possible families of designs,
all of which have variance $S^2$.
\label{set:pre_analysis}
\begin{itemize}
\item[(A)] Experiment with pre-analysis plan: 
As in Section~\ref{sec:design}, there is an unbiased design $E$.
Researchers cannot manipulate their findings, truthfully report $X = \theta + \varepsilon$, and pay a research cost $C_E$. 
The social planner chooses a publication rule $p_{E}^\star$ as in Definition \ref{defn:1} and incurs expected loss $\mathcal{L}_E^\star$. 
\item[(B)] Possible manipulation: We are in Setting~\ref{set:model_p_hacking}.  Thus, there is a family of designs $\Delta \in \mathbb{R}$.
The manipulation cost is $c_m < \infty$, and researchers can manipulate their findings after observing $\theta + \varepsilon$.
The fixed research cost is $C_0 = 0$.
The social planner chooses an optimal publication rule $p^\star$, and incurs a corresponding expected loss $\mathcal{L}_M^\star = \mathbb{E}_{\theta,\varepsilon}\left[\mathcal{L}_{p^\star}(X(\Delta_{p^\star}^{\star}), \Delta_{p^\star}^{\star}, \theta)\right]$ under the planner's preferred equilibrium. 
\end{itemize} 
\end{setting} 

In Scenario (A), researchers cannot manipulate their findings (as for, e.g., a pre-registered experiment), but pay a fixed cost $C_E$ of conducting an experiment.
In Scenario (B), researchers can manipulate their findings---as for an observational study.%
\footnote{While we assume zero research costs in Scenario (B), similar results hold for an manipulable studies that have sufficiently low research costs.}

To compare the planner's loss across the two scenarios, we introduce a quantitative measure of how much loss the expensiveness of the design $E$ in Scenario (A) entails for the planner. We call this measure the \emph{incentive cost} of $E$, as it captures the cost of incentizing the researcher to choose design $E$.
Namely, we consider the difference in the planner's loss (under the planner's optimal publication rules) between the design $E$ and a hypothetical design $E'$ with the same variance but no research costs.

\begin{defn}[Incentive costs] Given an unbiased design $E$, let $\mathrm{IC}(E) = \mathcal{L}_E^\star - \mathcal{L}_{E'}^\star$ where $E'$ is a cheap design with $C_{E'} = 0$ and variance $S_{E'}^2 = S_E^2$.  
\end{defn} 

Whenever $E$ is a cheap design, we have $\mathrm{IC}(E) =0$, while when $E$ is an expensive design, we have $\mathrm{IC}(E) >0$. Also, note that $\mathrm{IC}(E)$ is non-decreasing in $C_E$ by Lemma~\ref{lem:taylor_large} in Appendix~\ref{app:preliminary}, and can be readily computed by using the same lemma.

The next proposition provides a simple characterization of when the experiment in Scenario (A) is preferred by the planner to the observational study in Scenario (B) in terms of the incentive costs of the experiment in Scenario (A). High incentive costs for the experiment favor the observational study, and we provide a quantitative bound on incentive costs for observation studies to be preferred by the planner.

\begin{prop}[Experiment versus manipulable observational study] \label{prop:comparison_pap} 
In Setting~\ref{set:pre_analysis}:
\begin{itemize} 
\item[(a)] If $E$ is cheap (i.e., $\mathrm{IC}(E) = 0$), then 
$\mathcal{L}^\star_E < \mathcal{L}^\star_M$.
\item[(b)] If $\mathrm{IC}(E) > \frac{1 + 2 S c_m}{c_m^2}$, 
then $\mathcal{L}^\star_M < \mathcal{L}^\star_E$.
\end{itemize} 
\end{prop}

The proof is in Appendix~\ref{proof:pre_analysis}. When the cost of the experiment is high (and therefore $\mathrm{IC}(E)$ is large), whether the experiment or the observational study is preferred depends on on (i) the cost $C_E$ of the experiment and (ii) the cost of data manipulation $c_m$ in the observational study. 
Clearly, if $c_m$ is small, then an experiment with a pre-analysis plan may be preferred. Intuitively, with a low cost of manipulation, the planner must end up publishing a larger set of results that do not merit attention from the audience---thereby making the audience incur a possibly large attention cost. 

Now suppose that $c_m$ is large.  Then, an observational study with possible manipulation is preferred by the planner for sufficiently high research costs $C_E$ for the experiment---despite the cost of the experiment being private and paid only by the researcher. This preference arises because a sufficiently large experimental research cost $C_E$ increases the loss of the planner, who must publish results that do not merit the audience's attention.
This conclusion contrasts with analyses that abstract from the costs of experiments with pre-analysis plans, where (pre-registered) experiments always dominate (manipulable) observational studies \citep[e.g.,][]{kasy2023optimal, spiess2018optimal}.

\section{Application to medical studies} \label{Sec:application}

In this section, we study the implications of our model for medical studies.
We first calibrate our model from Section~\ref{sec:p_hacking} to meta-analyses of medical studies.
We then illustrate several properties of the optimal publication decision rules, and investigate when studies with potential manipulation may dominate experiments with no manipulation.

\subsection{Calibration}

To map the model in Section~\ref{sec:p_hacking} to the data,
we consider a family of studies $i$.
We normalize the estimates and estimands so that $S^2 = 1$.
Thus, we consider test statistics $X_i = \theta_i + \beta_i + \varepsilon_i$ where $\theta_i$ is the parameter of interest of the study $i$ defined as the size of the effect in standard deviation units, $\beta_i$ is a bias due to manipulation (potentially correlated with $\theta_i + \varepsilon_i$) and $\varepsilon_i \sim \mathcal{N}(0,1)$.\footnote{The normality assumption here captures a large-sample approximation.}
As in previous sections, we consider a prior $\theta_i \sim \mathcal{N}(0, \eta^2)$; here, the variance captures any additional prior heterogeneity arising from different studies, and the assumption of mean zero is consistent with the empirical distribution of studies in the Cochrane Database of Systematic Reviews, the leading database for systematic reviews in health care.\footnote{\url{https://www.cochranelibrary.com/cdsr/about-cdsr}, see \cite{bartovs2023empirical}.}

We use data from \cite{head2015extent} to calibrate our key parameters $(\eta^2, c_a, c_m)$
for manipulable designs.  As in Section~\ref{sec:preanalysisplans}, we impose $C_0 = 0$ in our model---i.e., that fixed costs are sunk at the stage at which researchers decide about whether/how to manipulate.%

\cite{head2015extent} collected $p$-values via text-mining the PubMed database across several disciplines. We focus here on $827,115$ $p$-values published in medical and pharmaceutical journals. 
As \citeauthor{head2015extent}'s \citeyearpar{head2015extent} data set does not contain $t$-statistics, we construct $t$-statistics by inverting a two-sided $t$-test.  That is, we construct $t$-statistics as $X_i:=\Phi^{-1}(1 - p_i/2)$, where $\Phi$ is the Gaussian CDF and $p_i$ is the $p$-value of study $i$.\footnote{Because the data does not contain information about $t$-tests directly, this assumes that the majority of $p$-values are constructed using two-sided $t$-tests. Although this may not always be the case, we note that $t$-tests are the standard practice in medical sciences \citep{fda}. It is possible to conduct our analysis assuming one-sided instead of two-sided tests.}
To adjust for publication bias in the data from \cite{head2015extent}, we use estimates from \cite{vorland2024publication} that $36\%$ of the medical studies are never published.  As a simplifying assumption, we assume that all such studies have non-significant results ($p$-value above $5\%$).

\paragraph{Calibration of $c_m$.} For competitive journals, standard $1$ or $5\%$ significance levels correspond to publication rules of the form $p^\star(X) = 1\{|X| \ge q\}$ for $q \in \{1.96,2.56\}$ where $q$ is the critical Gaussian quantile at a given $5\%$ or $1\%$ significance level.\footnote{Although some journals publish non-significant results, researchers may have incentives to manipulate their results to publish in top journals,
where most published results are significant \citep{laviolle2025trends}.}
We consider two calibrations that correspond to each of these two levels.
In each case, we would observe bunching at $q$, with researchers manipulating the findings in the interval $(q - \frac{1}{c_m}, q)$. Taking a standard bunching approach, we can therefore use the data to estimate $c_m$ by solving 
$\Phi(q) - \Phi(q - \frac{1}{c_m}) = b$, where $b$ is the share of findings in a small neighborhood around $q$. We find that about $18\%$ of findings have a $p$-value close to $0.05$ and $10\%$ of $p$-values close to $1\%$ in the set of studies reported by \cite{head2015extent}.\footnote{We choose $b$ as the share of $|X|$ between $1.95$ and $2$ for $q = 1.96$ and between $2.55$ and $2.6$ for $q = 2.56$. 
}
About $27\%$ of trials published in PubMed between 2002 and 2015 have some form of pre-registration \citep{lamberink2022clinical}. Assuming no manipulation of pre-registered findings, we estimate $b$ by multiplying these shares by $(1 - 0.36)/(1 - 0.27)$ to account for the $36\%$ of results that are not published, and $27\%$ of studies with pre-registration in PubMed. We find $c_m = 0.98$ for our 5\%-level calibration (which considers manipulation of a $t$-test of level~$5\%$), and $c_m = 0.83$ for our 1\%-level calibration (which considers manipulation of a $t$-test of level $1\%$). 

\paragraph{Calibration of $\eta^2$.} Manipulation may change the distribution of $X$, so we cannot directly estimate $\eta^2$ from second moments. However, for a publication rule of the form $1\{|X| \ge q\}$, under our model (from Setting~\ref{set:model_p_hacking} in Section~\ref{sec:p_hacking}), manipulation only occurs below the significance threshold. To construct a measure robust to manipulation, we take the $95$th percentile of the distribution of $X$, denoted as $\bar{q}_{95}$. Taking into account that an additional $36\%$ of findings is unobserved (unpublished) but non-significant by assumption, we can see that
$\bar{q}_{95} = \hat{q}_{95 - \gamma}$ where $\gamma = 0.05 \times 0.36/(1 - 0.36)$ and $\hat{q}$ is the empirical quantile of $t$-statistics from \cite{head2015extent}.  If $\hat{q}_{95 - \gamma} > 2.56$, our estimate will be robust to manipulation around the significance threshold under our model. We find that $\hat{q}_{95 - \gamma} = 3.43$---well above the critical values around which we expect manipulation. Using properties of the Gaussian distribution, 
$\eta^2 \approx (\bar{q}_{95}/2)^2 - 1$, which leads to an estimate of $\eta^2 = 1.94$.

\paragraph{Calibration of $c_a$.}
We choose $c_a$ to reflect the critical values of tests that are typically used in practice. As (absent manipulation) the critical threshold for publication is $\sqrt{c_a} \frac{1 + \eta^2}{\eta^2}$ in our model (Proposition~\ref{prop:t_test}), we consider take $\sqrt{c_a} \frac{1 + \eta^2}{\eta^2} = 1.96$ for our 5\%-level calibration (which captures a 5\%-level $t$-test), and $\sqrt{c_a} \frac{1 + \eta^2}{\eta^2} = 2.56$ for our 1\%-level calibration. 

\subsection{Optimal publication rules under \texorpdfstring{$p$}{p}-hacking}

\begin{figure}[t!]
    \centering
\includegraphics[scale = 0.65]{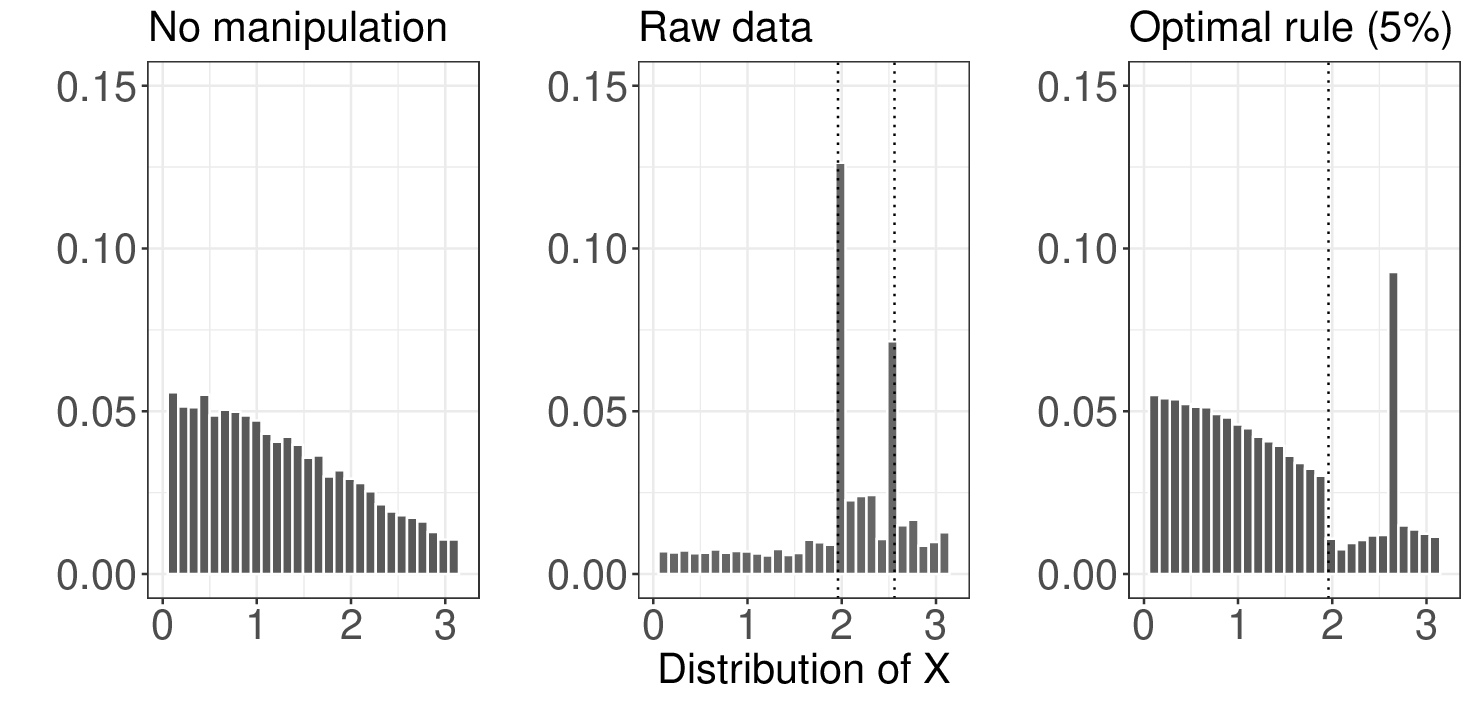}

    \caption{Simulated statistics without manipulation (left panel), empirical distribution of $t$-statistics in medical and pharmaceutical studies using data from \cite{head2015extent} (center panel), and simulated statistics with manipulation under the optimal publication rule with $c_a \frac{\eta^2 + 1}{\eta^2} = 1.96$ (right panel). Here, $\eta^2 = 1.94$ and $c_m = 0.98$ are calibrated to data from \cite{head2015extent}. In the center panel, we account for the $36\%$ of unobserved (insignificant) findings in reporting levels in the distribution. Dotted vertical lines correspond to values $1.96$ and $2.56$ in the center panel, and to value $1.96$ in the right~panel.}
    \label{fig:main}
\end{figure}

We next illustrate the impact of manipulation on the distribution of results, both in the data, and simulated in our model under the optimal publication rule.
In Figure \ref{fig:main}, we plot the distribution that we would see in the absence of manipulation in our calibration (left panel), the empirical distribution (center panel), and the distribution under the optimal publication rule under our 5\%-level calibration (right panel).
The empirical distribution exhibits bunching around $1.96$ and $2.56$, consistent with a much higher chance of publication above these significance thresholds.\footnote{In Figure~\ref{fig:raw} in Appendix~\ref{app:figs}, we report the raw distribution of $p$-values. Given rounding to two or three digit levels for $p$-values in the data, some of the bunching should be interpreted around (not exactly equal) to the critical thresholds \citep{elliott2022detecting}.}
The simulated distribution under the optimal publication rule features much less bunching.

\begin{table}[t!]
\centering
\small                                   
\setlength{\tabcolsep}{5pt}

\begin{tabular}{ccclc cc}
\toprule
\multirow{2}{*}{} &
\multirow{2}{*}{} &
\multirow{2}{*}{} &
\multirow{2}{*}{Publication rule} &
\multirow{2}{*}{\% Published} &
\multicolumn{2}{c}{Within published findings} \\
\cmidrule(lr){6-7}
 & & & & & \% Manipulated & Average Bias $|\beta|$ \\
\midrule
\multicolumn{7}{@{}l}{\textbf{5\%-Level Calibration:} $\,\eta^2 = 1.94,\;\sqrt{c_a}(1+\eta^{2})/\eta^{2}=1.96,\;c_m=0.98$}\\[-0.4ex]
\cmidrule(lr){1-7}
 & & & \makecell[l]{$t$-test rule $1\{|X| \ge 1.96\}$ \\(without manipulation)}  & 25\% & --- & --- \\ \cmidrule(lr){4-7}
 & & & \makecell[l]{$t$-test rule $1\{|X| \ge 1.96\}$ \\(with manipulation)}      & 58\% & 56\% & 0.31 \\ \cmidrule(lr){4-7}
 & & & Optimal rule ($X^{\star}=2.64$)                            & 25\% & 45\% &0.11 \\ \midrule
\addlinespace[0.8ex]
\multicolumn{7}{@{}l}{\textbf{1\%-Level Calibration:} $\,\eta^2 = 1.94,\;\sqrt{c_a}(1+\eta^{2})/\eta^{2}=2.56,\;c_m=0.83$}\\[-0.4ex]
\cmidrule(lr){1-7}
 & & & \makecell[l]{$t$-test rule $1\{|X| \ge 2.56\}$ \\(without manipulation)}  & 13\% & --- & --- \\ \cmidrule(lr){4-7}
 & & & \makecell[l]{$t$-test rule $1\{|X| \ge 2.56\}$ \\(with manipulation)}      & 42\% & 68\% & 0.45 \\ \cmidrule(lr){4-7}
 & & & Optimal rule: ($X^\star = 3.39$)                          & 13\% & 57\% & 0.18 \\
\bottomrule
\end{tabular}

 \caption{Properties of different publication rules by calibrating the parameters using data from \cite{head2015extent}. 
The two calibrations, which respectively capture 5\%-level and 1\%-level $t$-tests, have different choices of attention cost $c_a$ and manipulation cost $c_m$.
Within each calibration, the first row corresponds to the optimal publication rule assuming researchers do not manipulate the findings, which is a standard $t$-test rule. The second row shows equilibrium under that rule if researchers manipulate their findings. The last row for each calibration describes equilibrium under the optimal publication rule under manipulation.
  } 
  \label{tab:main} 
\end{table}

In Table \ref{tab:main}, we compare the simulated properties of the standard cutoff rule and the optimal rule in our calibrated model in more detail.
We calculate the cutoff $X^\star$ under the optimal publication rule, and compare the share of published findings and share of manipulated published findings between the optimal rule and the standard rules $t$-test rules $1\{|X| \ge t^*\}$ for $t^* \in \{1.96, 2.56\}$, both with and without manipulation. 

The first main takeaway is that the optimal publication rule substantially increases the threshold from $1.96$ to $2.64$ under the 5\%-level calibration (and from $2.56$ to $3.39$ under the 1\%-level calibration) due to the relatively small cost of manipulation. The new threshold $X^\star$ characterizes the point after which studies are published with probability 1. 

The second key observation is that under the optimal publication rule, the amount of manipulation within such studies is substantially lower---but non-zero. As Table~\ref{tab:main} shows, under the standard 5\%-level $t$-test rule $1\{|X| \ge 1.96\}$, about $54\%$ of the studies are published, and of these published findings the average bias is 0.31.
In contrast, if there were no manipulation, only $25\%$ of the findings would be published under the same threshold rule (this is because $X_i \sim \mathcal{N}(0, 2.94)$ unconditional on $\theta_i$). The optimal publication rule publishes a similar number of findings as in the absence of manipulation ($25\%$), and the average bias across all such published findings is $0.11$---less than half than under a standard $t$-test rule.

 Taken together, these results show that publication rules in leading medical journals should publish a similar number of findings if there was no manipulation under the standard 5\%-level $t$-test rule $1\{|X| \ge 1.96\}$. However, the critical threshold for which studies are published with probability 1 increases from 1.96 to $2.64$.  However, some results between 1.64 and 2.64 should still be published, albeit with probability less than 1. In particular, some results between 1.62 and 1.96 should be published to reduce incentives to substantially manipulate. Thus, publication decisions must depend on results through the ``publication~score'' 
$$
\begin{aligned} 
p^\star(X) = \begin{cases} 0 & \text{ if } X < X^\star - \frac{1}{c_m} \\ 1 - c_m (X^\star - |X|) & \text{ if } X^\star - \frac{1}{c_m} \le |X| \le X^\star \\ 
1 & \text{ if } X \ge X^\star 
\end{cases} 
\end{aligned} 
$$ 
where $(c_m, X^\star) = (0.98, 2.64)$, which indicates the relevance of findings for publication. 
For the 1\%-level case, one should instead take $(c_m, X^\star) = (0.83, 3.39)$ in this scoring rule.

\subsection{Pre-registered experiment versus manipulable observational study}  \label{sec:exp_application}

Returning to the exercise in Section~\ref{sec:preanalysisplans}, we ask when the planner should incentivize an experiment with costs $C_E > 0$ and no manipulation (e.g., adherence to a pre-analysis plan) over a manipulable observational study with lower research cost $C_0 = 0$.%
\footnote{%
The assumption that the observational study has no fixed research cost can be relaxed as long as its fixed costs do not affect the supply of research (individual rationality constraint is not binding).}
We suppose that the two studies share the same estimand,
and have equal precision absent bias from manipulation. 

In Figure \ref{fig:comparison1} (left panel), we report the ratio between the planner's optimal loss for the experiment and the loss for the observational study, as a function of the cost of the experiment $C_E$. We consider two publication rules for the observational study: the standard $t$-test rule that is optimal without manipulation \citep{frankel2022findings}, and the optimal publication rule accounting for manipulation (from Theorem~\ref{thm:optimal_assymetric_info}).

\begin{figure}[!t]
\centering 
\includegraphics[scale = 0.5]{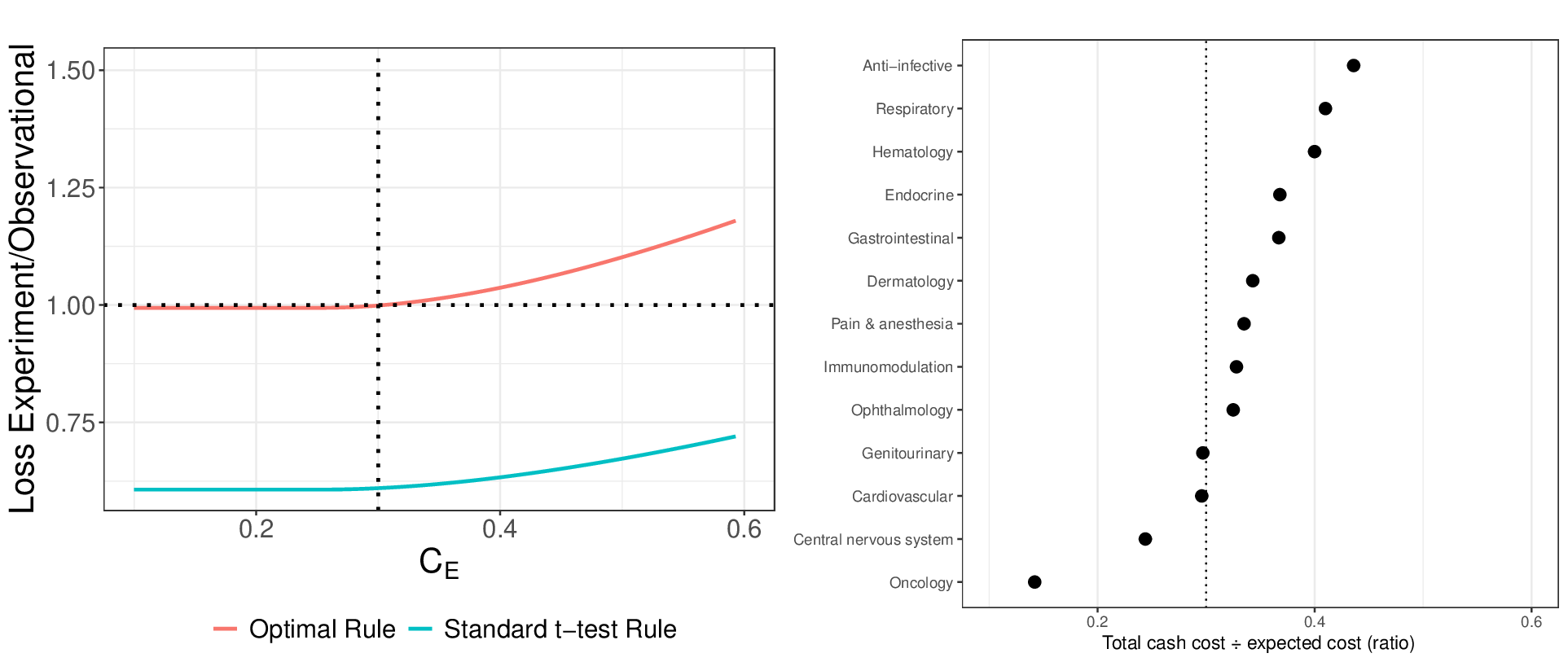}
\caption{The left panel reports the ratio between the loss of nonmanipulable experiment with varying costs $C_E$, and the loss of an inexpensive but manipulable observational study. Different colors give the losses of the observational study under a standard 5\%-level $t$-test rule $1\{|X| \ge 1.96\}$ (red), and under the estimated optimal publication rule for the 5\%-level calibration (blue). The right-hand side reports the ratio of the average cost of a Phase 3 trial across 13 pharmaceutical companies to expected capitalized costs (accounting for failures) from \cite{sertkaya2024costs}.} \label{fig:comparison1}
\end{figure} 

Under a standard $t$-test rule for the observational study, due to manipulability the observational study is dominated by the experiment, no matter the experimental cost. When we consider instead the optimal publication rule for the same observational study, 
the observational study performs much better.
Experiments with costs $C_E > 0.3$ (i.e., experiments whose research costs are above 30\% of the value of a publication) become dominated by observational studies.
This comparison reflects Proposition~\ref{prop:comparison_pap}, as high-cost experiments have high incentive costs.
Quantitatively, the underperformance of the observational study relative to the experiment is close to zero even for lower experimental research costs ($C_E < 0.3$): the expected loss from the experiment is more than $99\%$ of that of the observational study. 

To understand the magnitude of $C_E$ in different pharmaceutical industries, the right panel of Figure \ref{fig:comparison1} reports the ratio of total development cost of the drug to the expected costs (adjusting for cost of failures) using data from \citet[Columns 3 and 4 of Table 2]{sertkaya2024costs}. We interpret this as a proxy for $C_E$ for a researcher deciding whether to investigate the efficacy of a drug in a context with perfect competition (zero expected profits), which may (conservatively) capture private interests from drug companies when publication entails marketing and increased credibility \citep[e.g.][]{modi202310}.

We find substantial heterogeneity in whether an experiment or an observational study should be preferred under an optimal publication rule for the observational study. For most industries, $C_E$ is around $0.3$ or only slightly larger, suggesting that an experiment may be preferred. Three industries are however significantly above, with $C_E = 0.4$.
These results suggests that there may be welfare gains from allowing observational studies in industries in which experiments are particularly costly to implement---%
provided that publishers use different ``significance rules'' for observational studies that account for their manipulability.

\section{Conclusion}

This paper studies how researcher's incentives shape the optimal design of the scientific process.
Ignoring the researcher's incentives, it is optimal to publish the most surprising results \citep{frankel2022findings}.
When researcher's incentives matter, we show that optimal publication rules depend on private costs of research and incentives for research manipulation. 

In the absence of manipulation, we show that the planner prefers studies with larger mean-squared errors over sufficiently costly experiments. With manipulation, we show that it is optimal to (i) publish some unsurprising results and (ii) knowingly allow for manipulation at the margin. Observationally, the optimal policy would reduce the bunching of the findings around the publication cutoff. However, the optimal policy does not completely remove bunching. Even when the planner can eliminate manipulation by publishing only pre-specified experiments, this may not be the preferred policy when experiments entails large research costs.  In our application to medical studies, we highlight the importance of setting different publication rules for experimental and observational studies.

Our results have implications for the policy debate regarding the use of synthetical control groups in clinical trials.
Allowing synthetic  control groups opens the door to manipulation of the specification of the control group, similar to the manipulability of observational studies.
However, studies synthetic control groups are cheaper to implement as they can achieve similar power in a smaller experiment.
Our analysis suggests that when experiments are costly, allowing for synthetic controls may also be desirable---%
provided that publishers and approvers use different significance rules for them that account for their manipulability.

Future research should study more complex decisions by the planner and the researcher. For example, in contexts with pre-analysis plans, the planner may allow for the publication of non-prespecified findings. Finally, as our contribution lies at the intersection of econometrics and mechanism design, future research should study how other aspects of researchers' incentives impact the design of scientific communication.


\singlespacing
\bibliographystyle{chicago}
\bibliography{my_bib}

\clearpage

\appendix
\onehalfspacing

\section{Proofs omitted from Section~\ref{sec:design}}

Let $\phi$ denote the probability density function of a standard normal.

The following lemma provides a simple decomposition of the social planner's loss function conditional on the realized statistic $X$, which we use in several~proofs.

\begin{lem}[Loss function] \label{lem:loss} 
If $\Delta$ is an unbiased design, then we have
$$
\mathbb{E}\Big[\mathcal{L}_p\big(X(\Delta), \Delta, \theta\big)\Big| X(\Delta)\Big] = \Big[c_a - \frac{\eta^4}{(S_\Delta^2 + \eta^2)^2} X(\Delta)^2 \Big] p(X, \Delta) + \mathbb{E}\Big[\theta^2 \Big| X(\Delta)\Big]
$$
\end{lem} 
\begin{proof}
The lemma follows directly from the fact that $\mathbb{E}[\theta | X(\Delta)] = \frac{\eta^2}{S_\Delta^2 + \eta^2} X(\Delta)$. 
\end{proof} 

\subsection{Proof of Proposition \ref{prop:t_test} } \label{app:proof1}

The individual rationality constraint $\mathbb{E}[v_p(X(\Delta),\Delta)] \ge 0$ can be written equivalently as $\mathbb{E}[p(X(\Delta))] \ge C_\Delta$.
Note that the objective and constraints are linear in $p$.
Using Lemma \ref{lem:loss} and placing a Lagrange multiplier $\lambda \ge 0$ on the individual rationality constraint,
the Lagrangian of the planner's problem can be written as 
$$
\mathbb{E}\bigg[\bigg(c_a - \lambda - \Big(\frac{\eta^2}{S_\Delta^2 + \eta^2} X(\Delta) \Big)^2 \bigg) p\big(X(\Delta)\big)\bigg] + \lambda C_\Delta + \eta^2 
$$
We can solve and obtain $p_\lambda(X) = 1\big\{|X| \ge \frac{S_\Delta^2+\eta^2}{\eta^2} \sqrt{(c_a - \lambda)_+}\big\}$ as the minimizer of the Lagrangian.%
\footnote{Here, we write $a_+ = \max\{a,0\}$.}
%
As $X(\Delta) \sim \mathcal{N}(0, S_\Delta^2  + \eta^2)$,
the individual rationality constraint for $p_\lambda(X)$ is
$$
2\Phi\bigg(- \sqrt{(c_a - \lambda)_+} \frac{\sqrt{S_\Delta^2 + \eta^2}}{\eta^2 }\bigg) \ge C_\Delta \Longleftrightarrow \lambda \ge c_a - \bigg(\frac{\eta^2 }{\sqrt{S_\Delta^2 + \eta^2}}\big|\Phi^{-1}\big(\tfrac{C_\Delta}{2}\big)\big|\bigg)^2. 
$$
By complementary slackness, it follows that
\[\lambda = \max\bigg\{0,c_a - \bigg(\frac{\eta^2 }{\sqrt{S_\Delta^2 + \eta^2}}\big|\Phi^{-1}\big(\tfrac{C_\Delta}{2}\big)\big|\bigg)^2\bigg\}.\]
As we then have
\[\frac{S_\Delta^2+\eta^2}{\eta^2} \sqrt{(c_a - \lambda)_+} = \min\left\{\frac{S^2_\Delta + \eta^2}{\eta^2} \sqrt{c_a},\left|\Phi^{-1}(C_\Delta/2)\right| \sqrt{S_\Delta^2 + \eta^2}\right\} = t^\star_\Delta,\]
it follows that $p_\lambda(X) = p^\star_\Delta(X)$.

\subsection{Proof of Proposition \ref{prop:cheapWorthwhile}} \label{proof:prop:cheapWorthwhile}

Since $\Delta$ is cheap, Corollary~\ref{cor:t_testCheapExpensive} implies $t^\star_\Delta = \frac{S_\Delta^2 + \eta^2}{\eta^2}$.
Proposition~\ref{prop:t_test} and Lemma~\ref{lem:loss} then~imply
\[\mathcal{L}^\star_\Delta = \mathbb{E}\bigg[\Big(c_a - \frac{\eta^4}{(S_\Delta^2 + \eta^2)^2} X(\Delta)^2\Big) 1\Big\{|X(\Delta)| \ge \frac{S_\Delta^2 + \eta^2}{\eta^2} \Big\} \bigg] + \eta^2 < \eta^2.\]

\subsection{Preliminary Steps for Analysis of General Designs}
\label{app:preliminary}

Our analysis of general designs relies on the following lemma, which provides an exact characterization of the planner's loss under a threshold rule.

\begin{lem}
\label{lem:threshold}
If $\Delta$ is an unbiased design, then for any publication rule of the form $p(x) = 1\{|x| \ge t\},$ we have
\[\mathbb{E}\big[\mathcal{L}_{p}(X(\Delta),\Delta,\theta)\big] = \eta^2 + \mathbb{P}(|X(\Delta)| \ge t) c_a - \postvarred{\Delta}\Upsilon\big(\mathbb{P}(|X(\Delta)| \ge t)\big),\]
where we write $\Upsilon(t) = 2 \Phi^{-1}\big(1 - \tfrac{t}{2}\big)\phi\big(\Phi^{-1}\big(1 - \tfrac{t}{2}\big)\big) + t.$
\end{lem}
\begin{proof}
By Lemma~\ref{lem:loss}, we have
$$
\mathbb{E}\big[\mathcal{L}_{p}(X(\Delta), \Delta, \theta) \big| X(\Delta)\big]  = \mathbb{E}\big[\theta^2 \big| X(\Delta)\big] + c_a 1\big\{|X(\Delta)| \ge t\big\} - \Big(\frac{\eta^2 X(\Delta)}{S_\Delta^2 + \eta^2}\Big)^2 1\big\{|X(\Delta)| \ge t\big\}.
$$
Note that
\[
\mathbb{E}\left[ \Big(\frac{\eta^2 X(\Delta)}{S_\Delta^2 + \eta^2}\Big)^2 1\big\{|X(\Delta)| \ge t \big\} \right]
= \underbrace{\mathbb{E}\left[\Big(\frac{\eta^2 X(\Delta)}{S_\Delta^2 + \eta^2}\Big)^2 \,\bigg|\,  |X(\Delta)| \ge t \right]}_{(I)} \mathbb{P}\left(|X(\Delta)| \ge t \right).
\]
As $X(\Delta) \sim \mathcal{N}(0, S_\Delta^2 + \eta^2)$, it follows from the symmetry of the Gaussian distribution and properties of truncated Gaussian distributions that
$$
(I) = \mathbb{E}\left[\Big(\frac{\eta^2 X(\Delta)}{S_\Delta^2 + \eta^2}\Big)^2 \,\bigg|\, X(\Delta) \ge t\right] =  \frac{\eta^4}{S_\Delta^2 + \eta^2} \bigg[\frac{1}{\mathbb{P}\left(X(\Delta) \ge t  \right)} \frac{t}{\sqrt{S_\Delta^2+\eta^2}} \phi\Big(\frac{t}{\sqrt{S_\Delta^2+\eta^2}}\Big)+1\bigg]. 
$$

Taking expectations and using $\mathbb{P}(|X(\Delta)| \ge t) = 2 \mathbb{P}(X(\Delta) \ge t)$, it follows that
\begin{align*}
&\mathbb{E}\left[\mathcal{L}_{p}(X(\Delta), \Delta, \theta) \right]\\
&\quad = \eta^2 + c_a \mathbb{P}\left(|X(\Delta)| \ge t\right) - \frac{2\eta^4}{S_\Delta^2 + \eta^2} \frac{t}{\sqrt{S_\Delta^2+\eta^2}} \phi\Big(\frac{t}{\sqrt{S_\Delta^2+\eta^2}}\Big) - \frac{\eta^4}{S_\Delta^2 + \eta^2} \mathbb{P}\left(|X(\Delta)| \ge t  \right)\\
&\quad = \eta^2 + c_a \mathbb{P}\left(|X(\Delta)| \ge t  \right) - \postvarred{\Delta} \Upsilon\big(\mathbb{P}\left(|X(\Delta)| \ge t\right)\big),
\end{align*}
where the last equality holds as $X(\Delta) \sim \mathcal{N}(0,S_\Delta^2+\eta^2).$
\end{proof}

Using Lemma~\ref{lem:threshold}, we can characterize the planner's optimal loss.

\begin{lem}
\label{lem:taylor_large} 
If $\Delta$ is an unbiased design, then letting $\Upsilon$ be as in Lemma~\ref{lem:threshold}, we have
\[\mathcal{L}^\star_\Delta = \begin{cases}
\eta^2 + C_\Delta c_a - \postvarred{\Delta}\Upsilon(C_\Delta) &\text{if } \Delta \text{ is expensive}\\
\eta^2 + \mathbb{P}(|X(\Delta)| \ge \gamma^\star_\Delta) c_a - \postvarred{\Delta}\Upsilon(\mathbb{P}(|X(\Delta)| \ge \gamma^\star_\Delta)) &\text{if } \Delta \text{ is cheap}
\end{cases}.\]
Moreover, $\mathcal{L}^\star_\Delta$ is strictly decreasing in $\postvarred{\Delta}$ and non-decreasing in $C_\Delta$.
\end{lem} 
\begin{proof}
By Corollary~\ref{cor:t_testCheapExpensive}, we have $\mathbb{P}\left(|X(\Delta)| \ge t^\star_\Delta  \right) = \mathbb{P}\left(|X(\Delta)| \ge \gamma^\star_\Delta  \right)$ if $\Delta$ is cheap and $\mathbb{P}\left(|X(\Delta)| \ge t^\star_\Delta  \right) = C_\Delta$ if $\Delta$ is expensive.
Hence, the first part of the lemma follows from Lemma~\ref{lem:threshold}. 
Lemma~\ref{lem:threshold} also implies that for all $t \ge 0$, the loss $\mathbb{E}[\mathcal{L}_{1\{|x| \ge t\}}(X(\Delta),\Delta,\theta)]$ is strictly decreasing in $\postvarred{\Delta}$ and non-decreasing in $C_\Delta$.
As Proposition~\ref{prop:t_test} implies that
\[\mathcal{L}^\star_\Delta = \min_{t \ge 0} \big\{ \mathbb{E}[\mathcal{L}_{1\{|x| \ge t\}}(X(\Delta),\Delta,\theta)]\big\},\]
and the second part of the lemma follows.
\end{proof}

We next provide a few properties of the function $\Upsilon$ that will be used in the proofs.

\begin{lem} \label{lem:taylorUpsilon} 
Let $\Upsilon(t)$ be defined as in Lemma~\ref{lem:threshold}.
\begin{enumerate}[label=(\roman*)]
\item $\Upsilon(t)$ is strictly increasing on $(0,1)$ with derivative $\Upsilon'(t) = \Phi^{-1}\big(1 - \frac{t}{2}\big)^2$. \label{part:nondec}
\item $1 > \Upsilon(t) > 1 - (1-t)^3$ for all $0 < t < 1$. \label{part:upsilonBounds}
\item $\Upsilon(t) > 1 - \frac{1}{3}(1 - t) \Phi^{-1}\big(1 - \frac{t}{2}\big)^2$ for all $0 < t < 1$. \label{part:upsilonUglyLowerBound}
\end{enumerate} 
\end{lem}
\begin{proof}
Define $z(t) = \Phi^{-1}\!\big(1 - \tfrac{t}{2}\big)$, so that $\Upsilon(t) = 2z(t) \phi\big(z(t)\big) + t$.
By the chain rule,
we have
\[\Upsilon'(t) = 2z'(t) \phi\big(z(t)\big) + 2z(t) \phi'\big(z(t)\big) z'(t) + 1.\]
Recalling that $\phi'(x) = -x\phi(x)$, we have
\[\Upsilon'(t) = 2z'(t) \phi\big(z(t)\big) - 2z(t)^2 \phi\big(z(t)\big) z'(t) + 1 = 2 z'(t) \phi\big(z(t)\big)\big[1 - z(t)^2\big] + 1.\]
The chain rule also implies that
\begin{equation}
\label{eq:zPr}
z'(t) = -\frac{1}{2 \phi(z(t))},
\end{equation}
and it follows that
\begin{equation}
\label{eq:UpsilonPr}
\Upsilon'(t) = -\big(1 - z(t)^2\big) + 1 = z(t)^2.
\end{equation}
As $z(t) > 0$ for $t < 1$, Part~\ref{part:nondec} follows.
As $\Upsilon(1) = 1$, it follows that $\Upsilon(t) < 1$ for all $0 < t < 1$.

To bound $\Upsilon(t)$ from below for $t \ge \frac{1}{6}$, note that for $z = 0$, we have
\[\big(2\Phi(z) - 1\big)\sqrt{3} - z = 0.\]
For $z = \Phi^{-1}\big(\frac{11}{12}\big)$, noting that $\Phi^{-1}\big(\frac{11}{12}\big) < 1.4$, a simple numerical calculation shows that
\[\big(2\Phi(z) - 1\big)\sqrt{3} - z > \frac{5\sqrt{3}}{6} - 1.4 > 0.\] 
As $\phi$ is decreasing on $\mathbb{R}_{\ge 0}$, $\Phi$ is concave on $\mathbb{R}_{\ge 0}$.
It follows that for $0 < z \le \Phi^{-1}\big(\frac{11}{12}\big),$~we~have
\[\big(2\Phi(z)-1\big)\sqrt{3} > z.\]
Taking $z = z(t)$, it follows that
$(1-t)\sqrt{3} > z(t)$ for $\frac{1}{6} \le t < 1$.
In particular, for $\frac{1}{6} \le t < 1$, using \eqref{eq:UpsilonPr} yields that 
\[\Upsilon'(t) = z(t)^2 < 3(1-t)^2 = \frac{d}{dt} \Big(1 - (1-t)^3\Big).\]
Since $\Upsilon(1) = 1$, we have $\Upsilon(t) > 1 - (1-t)^3$ for $\frac{1}{6} \le t < 1$.

To bound $\Upsilon(t)$  from below for $0 < t \le \frac{1}{6}$,
note that $\Phi^{-1}\big(\frac{11}{12}\big) > 1$ and $\phi\big(\Phi^{-1}\big(\frac{11}{12}\big)\big) < \frac{1}{6}$ by simple numerical calculations.
As $\phi$ is decreasing on $\mathbb{R}_{\ge 0}$,
for $z \ge \Phi^{-1}\big(\frac{11}{12}\big)$, we then have
\[6\phi(z) \big(2\Phi(z) - 1\big) < \big(2\Phi(z) - 1\big) < 1 < z.\]
Taking $z = z(t)$, it follows that
$6 \phi(z(t)) (1-t) < z(t)$ for $0 < t \le \frac{1}{6}$.
Using \eqref{eq:zPr} then implies that $6(1-t) < -2z(t)z'(t)$ for $0 < t \le \frac{1}{6}$.
Differentiating \eqref{eq:UpsilonPr} then yields that for $0 < t \le \frac{1}{6}$, we have
\[\frac{d^2}{dt^2}\Big(\Upsilon(t) - \big(1 - (1-t)^3\big)\Big) = 2z(t)z'(t) + 6(1-t) < 0.\]
Hence, the function $\Upsilon(t) - \big(1 - (1-t)^3\big)$ is concave on $[0,\frac{1}{6}]$.
The previous paragraph showed that the function is positive at $t = \frac{1}{6}$, and a straightforward calculation shows that it vanishes at $t = 0$.
Hence, the function is positive for $0 < t \le \frac{1}{6}$, completing the proof of Part~\ref{part:upsilonBounds}.

To prove Part~\ref{part:upsilonUglyLowerBound},
note that as $\phi$ is decreasing on $\mathbb{R}_{\ge 0}$, we have that $\Phi(z) > \frac{1}{2} + z \phi(z)$ for $z > 0$.
Taking $z = z(t)$, it follows that for $0 < t < 1$, we have
\[\frac{1-t}{2} = \Phi\big(z(t)\big) - \frac{1}{2} > z(t) \phi\big(z(t)\big).\]
Using \eqref{eq:zPr}, we therefore have that $-(1-t)z'(t) > z(t)$ for $0 < t < 1$.
Using \eqref{eq:UpsilonPr} then yields that for $0 < t < 1$, we have
\[\Upsilon'(t) = z(t)^2 < \frac{1}{3} z(t)^2 - \frac{2}{3} (1-t) z(t) z'(t) = \frac{d}{dt} \Big(1 - \frac{1}{3}(1-t)z(t)^2\Big).\]
Since $\Upsilon(1) = 1$, Part~\ref{part:upsilonUglyLowerBound} follows by the definition of $z(t)$.
\end{proof}

\subsection{Proof of Proposition \ref{cor:worthwhileNeccSuff} } \label{proof:cor:worthwhileNeccSuff}

The following characterization of worthwhile designs is immediate from Lemma~\ref{lem:taylor_large}.

\begin{lem}
\label{lem:worthwhileNeccSuffUgly}
Let $\Upsilon(t)$ be as defined in Lemma~\ref{lem:taylor_large}.
An expensive design $\Delta$ is worthwhile if and only if
$\Upsilon(C_\Delta)\postvarred{\Delta} \ge C_\Delta c_a.$
\end{lem}

To derive Proposition~\ref{cor:worthwhileNeccSuff} from Lemma~\ref{lem:worthwhileNeccSuffUgly}, note that Lemma~\ref{lem:taylorUpsilon}\ref{part:upsilonBounds} implies that
\begin{align*}
\Upsilon(C_\Delta)\postvarred{\Delta} &\ge \postvarred{\Delta} \big(1 - (1-C_\Delta)^3\big)
\ge \postvarred{\Delta} - \eta^2(1-C_\Delta)^3\\
\Upsilon(C_\Delta)\postvarred{\Delta} &\le \postvarred{\Delta}
\end{align*}
where the second inequality on the first line holds as $\postvarred{\Delta} \le \eta^2$.

\subsection{Proof of Proposition \ref{prop:cheapPreferred}} \label{proof:prop:cheapPreferred}

Consider an unbiased design $E'$ with $S_{E'} = S_E$ and $C_{E'} = 0$.
As $E'$ is also cheap,
Lemma~\ref{lem:taylor_large} implies $\mathcal{L}^\star_{E'} = \mathcal{L}^\star_E$.
Since $S_E^2 \le S_O^2$ 
(in Setting~\ref{set:1}),
we have $\postvarred{E'} = \postvarred{E} > \postvarred{O}$.
Hence, Lemma~\ref{lem:taylor_large} implies $\mathcal{L}^\star_{E'} \le \mathcal{L}^\star_O$.
It follows that $\mathcal{L}^\star_E = \mathcal{L}^\star_{E'} \le \mathcal{L}^\star_O$.


\subsection{Proof of Proposition \ref{prop:cpPreferred}}  \label{proof:prop:cpPreferred}

It follows from Lemma~\ref{lem:taylor_large} that
\begin{multline*}
\mathcal{L}_E^\star - \mathcal{L}_O^\star = (C_E - C_O) c_a - \big(\postvarred{E} - \postvarred{O}\big)\\
+ \postvarred{E}(1 -\Upsilon(E)) - \postvarred{O}(1 -\Upsilon(O)).
\end{multline*}
Hence, $E$ is planner-preferred to $O$ if and only if $c_a < c_a^\star(E,O,\eta),$ where
\begin{multline*}
c_a^\star(E,O,\eta) = \frac{\postvarred{E} - \postvarred{O}}{C_E - C_O}\\
- \frac{\postvarred{E}(1 -\Upsilon(E)) - \postvarred{O}(1 -\Upsilon(O))}{C_E-C_O}.
\end{multline*}
The proposition follows
as we have that $0 \le \postvarred{\Delta} \le \eta^2$ (by construction)
and $1-\Upsilon(\Delta) = O((1-C_\Delta)^3)$ (by Lemma~\ref{lem:taylorUpsilon}\ref{part:upsilonBounds}) for all designs $\Delta$.

\subsection{Proof of Theorem \ref{prop:main_comparison} } \label{proof:prop:main_comparison}


We first consider the case in which $O$ is expensive, and apply that case to deduce the result for the case in which $O$ is cheap.  
When $O$ is expensive, by Lemma~\ref{lem:taylor_large}, we can write 
\begin{equation}
\label{eq:lossDiffExpensive}
\mathcal{L}_E^\star - \mathcal{L}_O^\star = C_Ec_a - \postvarred{E} \Upsilon(C_E) - C_O c_a + \postvarred{O} \Upsilon(C_O).
\end{equation} 

\paragraph*{Proof of (a)}
By Lemma~\ref{lem:taylorUpsilon}\ref{part:nondec}, since $C_E > C_O$ 
(in Setting~\ref{set:1}),
it follows from \eqref{eq:lossDiffExpensive} that
\[\mathcal{L}_E^\star - \mathcal{L}_O^\star < (C_E - C_O) c_a - \big(\postvarred{E} - \postvarred{O}\big) \Upsilon(C_E).\]
Note that as $C_E \in [0,1]$, Lemma~\ref{lem:taylorUpsilon}\ref{part:upsilonBounds} implies that
\[\Upsilon(C_E) \ge 1 - (1-C_E)^3 \ge 1 - (1-C_E) = C_E.\]
Since $\postvarred{E} \ge \postvarred{O}$ 
(in Setting~\ref{set:1}),
it follows that
\[\mathcal{L}_E^\star - \mathcal{L}_O^\star < (C_E - C_O) c_a - \big(\postvarred{E} - \postvarred{O}\big) C_E \le 0.\]
%
%

\paragraph*{Proof of (b)}
As $C_O < C_E$ (in Setting~\ref{set:1}) and $C_E \le 1$, we have $C_O < 1$.
Hence,
by Parts~\ref{part:upsilonBounds} and~\ref{part:upsilonUglyLowerBound} of Lemma~\ref{lem:taylorUpsilon}, it follows from \eqref{eq:lossDiffExpensive} that
\begin{multline*}
\mathcal{L}_E^\star - \mathcal{L}_O^\star > (C_E - C_O) c_a - \postvarred{E}\\
+ \postvarred{O} - \frac{1}{3}\postvarred{O} (1-C_O) \Phi^{-1}\Big(1 - \frac{C_O}{2}\Big)^2.
\end{multline*}

As $O$ is expensive and $X(O) \sim \mathcal{N}(0, S_O^2 + \eta^2)$, we have
\[C_O \ge \mathbb{P}(|X(O)| \ge \gamma_O^\star) =  2\bigg[1 - \Phi\bigg(\sqrt{\frac{c_a}{\postvarred{O}}}\bigg)\bigg].\]
Rearranging terms and applying $\Phi^{-1}$ yields that
\[\Phi^{-1}\Big(1 - \frac{C_O}{2}\Big)^2 \le \frac{c_a}{\postvarred{O}}\]

Hence, we have that
\begin{align*}
\mathcal{L}_E^\star - \mathcal{L}_O^\star &> (C_E - C_O) c_a - \postvarred{E} + \postvarred{O} - \frac{(1-C_O)c_a}{3}\\
&= \Big(C_E - \frac{2C_O+1}{3}\Big) c_a - \Big( \postvarred{E} - \postvarred{O}\Big) \ge 0.
\end{align*}

\paragraph*{Proof for the case in which $O$ is cheap}
Consider an unbiased design $O'$ with $S_{O'} = S_O$ and $C_{O'} = \mathbb{P}(|X(O)| \ge \gamma_O^\star)).$
By Lemma~\ref{lem:taylor_large}, we have that $\mathcal{L}^\star_{O'} = \mathcal{L}^\star_{O}$.
Moreover, $O'$ is expensive.
Hence, we can conclude the assertions for the comparison between $E$ and $O$ by applying Parts (a) and (b) of the theorem to compare $E$ to the expensive design $O'$.

\section{Proofs omitted from Section~\ref{sec:p_hacking}}
\label{app:proofsPhacking}

\newcommand\type{Y}

Let $\type = \theta + \varepsilon$ denote the type of the researcher 
in Setting~\ref{set:model_p_hacking}.
Let $\omega = \frac{\eta^2}{S^2 + \eta^2}$ denote the proportional posterior variance reduction from publishing any design in Setting~\ref{set:model_p_hacking}. 

We begin by characterizing the planner's expected loss conditional on $Y$.

\begin{lem}
\label{lem:lossThmAsymmetric}
In Setting~\ref{set:model_p_hacking}, writing $\omega = \frac{\eta^2}{S^2 + \eta^2}$:
\begin{enumerate}[label=(\alph*)]
\item The planner's expected loss conditional on $\type$ if the researcher chooses a design $\Delta$ that has publication probability $p$ (conditional on $\type$) is
\[\left(\omega^2 (\beta_\Delta^2 - \type^2)+ c_a\right)p + \omega^2 \type^2 + \omega S^2.\]
\item If no design is chosen,
then the planner's expected loss conditional on $\type$ is $\omega^2 \type^2 + \omega S^2.$
\end{enumerate}
\end{lem}
\begin{proof}
Note that $\theta | \type \sim \mathcal{N}(\omega \type,\omega S^2)$.
The audience's action is $a^\star_p(X,\Delta) = \omega (\type + \beta_\Delta)$ if $X$ is published and $a_p^\star = 0$ if no results are published.
Hence, the expected loss conditional on $\type$ if the researcher chooses a design $\Delta$ that has publication probability $p$ (conditional on $\type$) is
\begin{equation*}
\left(\omega^2 \beta_\Delta^2 + c_a\right)p + \omega^2 \type^2 (1-p) + \omega S^2 = \left(\omega^2 (\beta_\Delta^2 - \type^2)+ c_a\right)p + \omega^2 \type^2 + \omega S^2,
\end{equation*}
which proves Part~(a).
Part~(b) follows directly from the fact that $\theta | \type \sim \mathcal{N}(\omega \type,\omega S^2)$.
\end{proof}

Motivated by Lemma~\ref{lem:lossThmAsymmetric}, for $0 \le v \le 1$,
consider the function
\begin{equation}
\label{eq:lstar}
\mathcal{L}^\star(\type,v) = \min_{p,\beta_\Delta \in [0,1] \times \mathbb{R} \big| p - c_m |\beta_\Delta| = v} \Big\{\left(\omega^2 (\beta_\Delta^2 - \type^2)+ c_a\right)p\Big\}.
\end{equation}
By Lemma~\ref{lem:lossThmAsymmetric}, 
in Setting~\ref{set:model_p_hacking},
$\mathcal{L}^\star(\type,v)$
gives the difference between the minimum expected loss conditional on $\type$ generated by a design $\Delta$ and publication probability that delivers utility $v-C_0$ to type $\type$,
and the conditional expected loss if no results are published.
We next derive two properties of $\mathcal{L}^\star(\type,v)$ and the optimizer in \eqref{eq:lstar}.
Here, we let $\gamma^\star = \frac{S^2 + \eta^2}{\eta^2} \sqrt{c_a}$.

\begin{lem}
\label{lem:asymmetrichelp}
\begin{enumerate}[label=(\alph*)]
\item \label{part:asymmetrichelpunique} If $|\type| \not= \gamma^\star$, then there is a unique optimizer $(\tilde{p}(\type,v),\beta_\Delta(\type,v))$ in \eqref{eq:lstar}, which is given by
\[\tilde{p}(\type,v) = \begin{cases}
v & \text{if } |\type| \le \gamma^\star\\
\min\left\{1,\frac{2v + \sqrt{v^2 + 3 c_m^2 (\type^2 - (\gamma^\star)^2)}}{3}\right\} & \text{if } |\type| > \gamma^\star
\end{cases}\]
and $\beta_\Delta(\type,v) = \frac{\tilde{p}(\type,v) - v}{c_m}$.
\item \label{part:asymmetrichelphardmonotonev} If $|\type| > \gamma^\star$ (resp.~$|\type| = \gamma^\star$, $|\type| < \gamma^\star$),
then
$\mathcal{L^\star}(\type,v)$ 
is negative (resp.~zero, positive) for $v > 0$,
and has negative (resp.~zero, positive) derivative in $v$ over $(0,1)$.
\end{enumerate}
\end{lem}
\begin{proof}
Writing $|\beta_\Delta| = \frac{p-v}{c_m}$,
note that \eqref{eq:lstar} can be written equivalently as
\begin{equation}
\label{eq:lossUtilSubstituted}
\mathcal{L^\star}(\type,v) = \min_{p \in [0,1] \mid p \ge v} \left\{ \left[\omega^2 \left(\frac{p-v}{c_m}\right)^2 - \omega^2 \type^2 + c_a\right] p \right\}.
\end{equation}
Noting that $\gamma^\star = \frac{1}{\omega} \sqrt{c_a}$,
we divide into cases based on the value of $|\type|$ to complete the~proof.
\begin{itemize}
    \item Case 1: $|\type|>\gamma^\star$.
    In this case, we claim that for $0 < v < 1$,
    all optimizers $\check{p}$ in the relaxed problem
    \begin{equation}
    \label{eq:lossUtilSubstitutedRelaxed}
    \min_{p \in [0,1]} \left\{ \left[\omega^2 \left(\frac{p-v}{c_m}\right)^2 - \omega^2 \type^2 + c_a\right] p \right\},
    \end{equation}
    satisfy $\check{p} > v$.
    Let $\check{p}$ be any optimizer in \eqref{eq:lossUtilSubstitutedRelaxed}.
    As the objective in \eqref{eq:lossUtilSubstitutedRelaxed} is negative at $p = v$, the optimum is negative. It follows that $\check{p} > 0$ and $\omega^2\left(\frac{\check{p}-v}{c_m}\right)^2 - \omega^2 \type^2 + c_a < 0$.
    The first-order condition then entails that $\check{p} = 1$ or
    \[0 = \frac{2 \omega^2(\check{p}-v)}{c_m} \check{p} + \left[\omega^2 \left(\frac{\check{p}-v}{c_m}\right)^2 - \omega^2 \type^2 + c_a\right] < \frac{2 \omega^2(\check{p}-v)}{c_m} \check{p},\]
    and it follows that $\check{p} > v$, as desired.
    Hence, the problems \eqref{eq:lossUtilSubstituted} and \eqref{eq:lossUtilSubstitutedRelaxed} are equivalent.
    
    It also follows from first-order condition that
    \[\check{p} = \min\left\{1,\frac{2v + \sqrt{v^2 + 3 c_m^2 (\type^2 - (\gamma^\star)^2)}}{3}\right\}\]
    is the unique optimizer in \eqref{eq:lossUtilSubstitutedRelaxed}.
    The case of Part~\ref{part:asymmetrichelpunique} of the lemma for $|Y| > \gamma^\star$ follows.

    To prove the analogous case of Part~\ref{part:asymmetrichelphardmonotonev}, note that as \eqref{eq:lossUtilSubstitutedRelaxed} has a negative optimum, we have $\mathcal{L}^\star(\type;v) < 0$.    
    The Envelope Theorem \citep[Theorem 3]{milgrom2002envelope} applied to \eqref{eq:lossUtilSubstitutedRelaxed} guarantees that $\mathcal{L}^\star(\type;v)$ is partially differentiable in $v$
    , and that
    $$\frac{\partial \mathcal{L}^\star(\type,v)}{\partial v} = \frac{2 \omega^2 (\check{p} - v\big) \check{p}}{c_m^2} < 0.$$
    
    \item Case 2: $|\type| \le \gamma^\star$.  In this case, the objective in \eqref{eq:lossUtilSubstituted} is increasing in $p$ on the interval $[v,1]$.  The optimum is therefore achieved at $\check{p} = v$ (uniquely for $|\type| < \gamma^\star$),
    so we have that
    $\mathcal{L^\star}(\type,v) = (c_a - \omega^2 \type^2) v$.
    For $|\type| = \gamma^\star$, this function is zero.
    For $|\type| < \gamma^\star$, this function is positive for $v > 0$, and has positive derivative in $v$. 
\end{itemize}
The cases exhaust all possibilities, which completes the proof.
\end{proof}

\subsection{Proof of Theorem~\ref{thm:optimal_assymetric_info}} \label{proof:thm:1}


Consider first the constrained problem in which the planner must choose a publication rule that provides type $ \gamma^\star$ an indirect utility of $u^\star \in [0,1-C_0]$.%
\footnote{Recall that indirect utility refers to the type's utility under any utility-maximizing design choice.}
The linearly smoothed cutoff rule $p^\star_{\gamma^\star + \frac{1-C_0-u^\star}{c_m},c_m}$ lies within this class.
In fact, we can use Lemma~\ref{lem:lossThmAsymmetric} to show that this publication rule is optimal, and essentially uniquely so, in the constrained problem.

\begin{lem}
\label{lem:asymmetrichelplinearsmooth}
Within the class of publication rules that provide type  $\gamma^\star$ an indirect utility of $u^\star \in [0,1-C_0]$,
for all types $\type > 0$:
\begin{enumerate}[label=(\alph*)]
    \item \label{part:asymmetrichelplinearsmoothopt} the linearly smoothed cutoff rule $p_{\gamma^\star + \frac{1-C_0-u^\star}{c_m},c_m}$ minimizes the expected loss conditional on $\type$ (under the planner's preferred equilibrium), and
    \item \label{part:asymmetrichelplinearsmoothunique} any other publication rule within this class minimizes the expected loss conditional on $\type$ must provide the same indirect utility to type $\type$ as $p_{\gamma^\star + \frac{1-C_0-u^\star}{c_m},c_m}$.
\end{enumerate}
\end{lem}
\begin{proof}
Under the planner's preferred equilibrium, the publication rule $p_{\gamma^\star + \frac{1-C_0-u^\star}{c_m},c_m}$ delivers expected loss conditional on $\type$ given by
\[
\begin{cases}
\omega^2 \type^2 + \omega S^2 & \text{if } |\type| \le \gamma^\star - \frac{u^\star}{c_m}\\
\mathcal{L}^\star(\type,u^\star + c_m (\type - \gamma^\star) + C_0) + \omega^2 \type^2 + \omega S^2 & \text{if } \gamma^\star - \frac{u^\star}{c_m} < |\type| < \gamma^\star + \frac{1-C_0-u^\star}{c_m}\\
\mathcal{L}^\star(\type,1) + \omega^2 \type^2 + \omega S^2 & \text{if } |\type| \ge \gamma^\star + \frac{1-C_0-u^\star}{c_m}
\end{cases}.
\]
We divide into cases based on the value of $\type$ to prove the lemma.
\begin{itemize}
    \item Case 1: $\type > \gamma^\star$.  If type $\type$ obtains utility $u$, then type $\gamma^\star$ could obtain utility at least $u - c_m (\type - \gamma^\star)$ by choosing a design to obtain the same results as type $\type$.
    Hence, we must have that $u - c_m (\type - \gamma^\star) \le u^\star$.
    As $C_\Delta \ge C_0$, we have $u \le 1 - C_0$.
    By Lemma~\ref{lem:asymmetrichelp}\ref{part:asymmetrichelphardmonotonev} and the definition of $\mathcal{L}^\star$, the expected loss conditional on $\type$ must be at least \[\mathcal{L}^\star\big(\type,\min\big\{u^\star + c_m (\type - \gamma^\star) + C_0,1\big\}\big) + \omega^2 \type^2 + \omega S^2 \]
    with equality only if the indirect utility of type $\type$ is $\min\big\{u^\star + c_m (\type - \gamma^\star) ,1 - C_0\big\}.$
    \item Case 2: $\type = \gamma^\star$.  By Lemma~\ref{lem:asymmetrichelp}\ref{part:asymmetrichelphardmonotonev} and the definition of $\mathcal{L}^\star$, the expected loss conditional on $\type$ must be at least $\omega^2 \type^2 + \omega S^2 = \mathcal{L}^\star(\type,u^\star + C_0) + \omega^2 \type^2 + \omega S^2$.
    \item Case 3: $\gamma^\star - \frac{u^\star}{c_m} \le \type < \gamma^\star$.  Type $\type$ could obtain utility at least $u^\star + c_m (\type - \gamma^\star) \ge 0$ by choosing a design to obtain the same results as type $\gamma^\star$.
    By Lemma~\ref{lem:asymmetrichelp}\ref{part:asymmetrichelphardmonotonev} and the definition of $\mathcal{L}^\star$, the expected loss conditional on $\type$ must be at least \[\mathcal{L}^\star(\type,u^\star + c_m (\type - \gamma^\star) + C_0) + \omega^2 \type^2 + \omega S^2,\]
    with equality only if the indirect utility of type $\type$ is $u^\star + c_m (\type - \gamma^\star)$.
    \item Case 4: $\type < \gamma^\star - \frac{u^\star}{c_m}$.
    By Lemma~\ref{lem:asymmetrichelp}\ref{part:asymmetrichelphardmonotonev} and the definition of $\mathcal{L}^\star$, the expected loss conditional on $\type$ must be at least $\omega^2 \type^2 + \omega S^2$,
    with equality only if the indirect utility of type $\type$ is 0.
\end{itemize}
The cases exhaust all possibilities, which completes the proof.
\end{proof}

We are now ready to complete the proof of the theorem.

\paragraph*{Proof of Part~(a)} Lemma~\ref{lem:asymmetrichelplinearsmooth}\ref{part:asymmetrichelplinearsmoothopt} implies there exists a utility level $u^\star \in [0,1-C_0]$ for type $\gamma^\star$ such that $p_{\gamma^\star + \frac{1-C_0-u^\star}{c_m},c_m}$ (under the planner's preferred equilibrium) is optimal.
The expected loss of $p_{\gamma^\star + \frac{1-C_0-u^\star}{c_m},c_m}$ (under the planner's preferred equilibrium) is
\begin{equation}
\label{eq:optXAsymmetricHelp}
\mathcal{E}(u^\star,C_0) = \mathbb{E}_{\type}\left[\mathcal{L}^\star\big(\type,\min\{u^\star+C_0+c_m(\type-\gamma^\star),1\}\big) \, 1\left\{|\type| \ge \gamma^\star - \frac{u^\star}{c_m}\right\}\right] + \eta^2.
\end{equation}
Differentiating under the integral sign using Lemma~\ref{lem:asymmetrichelp}\ref{part:asymmetrichelphardmonotonev},
we see that $\frac{\partial \mathcal{E}}{\partial u^\star} \big|_{u^\star = 0} < 0$ and that $\frac{\partial \mathcal{E}}{\partial u^\star} \big|_{u^\star = 1-C_0} > 0$.
So for $p_{\gamma^\star + \frac{1-C_0-u^\star}{c_m},c_m}$ to be optimal,
we must have $0 < u^\star < 1-C_0$, and hence that
$\gamma^\star < \gamma^\star + \frac{1-C_0-u^\star}{c_m} < \gamma^\star + \frac{1 - C_0}{c_m}.$

\paragraph*{Proof of Part (b)}
Consider an optimal publication rule $p^\star$
that delivers indirect utility $u^\star$ to type $\gamma^\star$.
By Lemma~\ref{lem:asymmetrichelplinearsmooth},
the publication rule $p^\star$ (under the planner's preferred equilibrium) must deliver expected loss conditional on $\type$ equal to that of $p_{\gamma^\star + \frac{1-C_0-u^\star}{c_m},c_m}$ for almost all types $\type > 0$,
and $p_{\gamma^\star + \frac{1-C_0-u^\star}{c_m},c_m}$ must be optimal.
In particular, we have $0 < u^\star < 1 - C_0$.
Using these consequences of optimality,
we prove the two assertions in this part separately.

First, suppose for sake of deriving a contradiction that $p^\star(X) \not= p_{\gamma^\star + \frac{1-C_0-u^\star}{c_m},c_m}(X)$ for a positive measure of $X > 0$ with $p_{\gamma^\star + \frac{1-C_0-u^\star}{c_m},c_m}(X) > C_0$.
Then, at least one of the following must occur.
\begin{itemize}
    \item Case 1: $p^\star(X) \not= 1$ for a positive measure of results $X > \gamma^\star + \frac{1 - C_0-u^\star}{c_m}$.
    Then types $\type > \gamma^\star + \frac{1 - C_0-u^\star}{c_m}$ with $p^\star(\type) < 1$ must obtain indirect utility less than $1 - C_0$,
    which, by Lemma~\ref{lem:asymmetrichelplinearsmooth}\ref{part:asymmetrichelplinearsmoothunique}, must lead to expected loss conditional on $\type$ strictly greater than that of $p_{\gamma^\star + \frac{1-C_0-u^\star}{c_m},c_m}$---a contradiction.
    \item Case 2: the set $A =\left\{X \in \left(\gamma^\star - \frac{u^\star}{c_m}, \gamma^\star + \frac{1 - C_0-u^\star}{c_m}\right) \Big| p^\star(X) \not= u^\star + C_0 + c_m(X - \gamma^\star)\right\}$
    has positive measure.
    As $p$ is assumed Borel measurable, the set $A$ is a Borel set.
    Letting $\tilde{p}(\type,v)$ be as defined in Lemma~\ref{lem:asymmetrichelp}\ref{part:asymmetrichelpunique},
    $\tilde{p}(\type,u^\star + C_0 +  c_m(\type - \gamma^\star))$ is absolutely continuous in $\type$ and evaluates to $C_0$ (resp.~1) at $\type = \gamma^\star - \tfrac{u^\star}{c_m}$ (resp.~$\type = \gamma^\star + \frac{1-C_0-u^\star}{c_m}$).
    In particular, by the Intermediate Value Theorem, $\tilde{p}(\type,u^\star + C_0 +  c_m(\type - \gamma^\star))$ is surjective from $\left(\gamma^\star - \tfrac{u^\star}{c_m}, \gamma^\star + \tfrac{1 - C_0-u^\star}{c_m}\right)$ to $\left(\gamma^\star - \tfrac{u^\star}{c_m}, \gamma^\star + \tfrac{1 - C_0-u^\star}{c_m}\right).$
    Hence, the set
    \[A' = \left\{Y \in \left(\gamma^\star - \tfrac{u^\star}{c_m}, \gamma^\star + \tfrac{1 - C_0-u^\star}{c_m}\right) \bigg| \type + \tfrac{\tilde{p}(\type,u^\star + C_0 +  c_m(\type - \gamma^\star)) - u^\star - C_0 - c_m(\type - \gamma^\star)}{c_m} \in A\right\}\]
    has positive measure.

    Lemma~\ref{lem:asymmetrichelplinearsmooth}\ref{part:asymmetrichelplinearsmoothunique} implies for $\type \in \left(\gamma^\star - \frac{u^\star}{c_m}, \gamma^\star + \frac{1 - C_0-u^\star}{c_m}\right)$,
    to obtain expected loss conditional on $\type$ equal to that of $p_{\gamma^\star + \frac{1-C_0-u^\star}{c_m},c_m}$,
    type $\type$ must obtain indirect utility $u^\star +  c_m(\type - \gamma^\star)$.
    But if a type $\type \in A'$ obtains indirect utility $u^\star +  c_m(\type - \gamma^\star)$,
    they must have publication probability different from $\tilde{p}(X,u^\star + C_0 +  c_m(\type - \gamma^\star))$. By Lemma~\ref{lem:asymmetrichelp}\ref{part:asymmetrichelpunique}, this would lead to expected loss conditional on $\type$ strictly greater than under $p_{\gamma^\star + \frac{1-C_0-u^\star}{c_m},c_m}$---which is a contradiction as $A'$ has positive measure.
\end{itemize}

Next, suppose for sake of deriving a contradiction that $p^\star(X) > C_0$ for a positive measure of $X > 0$ with $p_{\gamma^\star + \frac{1-C_0-u^\star}{c_m},c_m}(X) \le C_0$.
Then types $\type \in \big(0,\gamma^\star - \frac{1-u^\star}{c_m}\big)$ with $p^\star(\type) > C_0$ must obtain positive indirect utility,
which, by Lemma~\ref{lem:asymmetrichelplinearsmooth}\ref{part:asymmetrichelplinearsmoothunique}, must lead to expected loss conditional on $\type$ strictly greater than that of $p_{\gamma^\star + \frac{1-C_0-u^\star}{c_m},c_m}$---a contradiction.

\subsection{Proof of Propositions~\ref{prop:non_surprisingPublish},~\ref{prop:surprisingNotPublish},~\ref{prop:manipulation_equilibrium}, and~\ref{prop:bunching_equilibrium}}
\label{proof:prop:non_surprising}

Let $X^\star$ be as in Theorem~\ref{thm:optimal_assymetric_info}(b), and
%
let $\type \in (X^\star - \frac{1-C_0}{c_m},X^\star)$.

Under publication rule $p^\star$, type $\type$ is indifferent between all designs $\Delta$ with $0 \le \beta_\Delta \le X^\star - \type$,
and strictly prefers such designs to all other designs.
These designs yield utility $v - C_0$, where $v = 1 - c_m(X^\star - \type) > C_0$.
Optimality then implies that almost surely, the lowest loss design $\Delta$ with $0 \le \beta_\Delta \le X^\star - \type$ will be chosen.
Lemma~\ref{lem:lossThmAsymmetric} implies that the expected loss conditional on $\type$ if type $\type$ chooses such a design $\Delta$ is
\begin{equation*}
\left(\omega^2 (\beta_\Delta^2 - \type^2)+ c_a\right)p + \omega^2 \type^2 + \omega S^2 \quad \text{where} \quad p = 1 - c_m(X^\star - \type - \beta_\Delta).
\end{equation*}
Hence, defining $\tilde{p}(\type,v)$ and $\beta_\Delta(\type,v)$ as in Lemma~\ref{lem:asymmetrichelp}\ref{part:asymmetrichelpunique}, the loss-minimizing design has bias $\beta_\Delta(\type,v)$,
and the publication probability will be $p^\star(X) = \tilde{p}(\type,v)$.

To prove Proposition~\ref{prop:non_surprisingPublish},
suppose that $\type \in \big(X^\star-\frac{1-C_0}{c_m},\gamma^\star\big)$.
Lemma~\ref{lem:asymmetrichelp}\ref{part:asymmetrichelpunique}
implies that $\tilde{p}(\type,v)= v$ and $\beta_\Delta(\type,v) = 0$.
Hence, almost surely, we have $p^\star(X) > C_0$ and $\beta_{\Delta^\star_{p^\star}} = 0$.

To prove Proposition~\ref{prop:surprisingNotPublish},
suppose that $\type > \gamma^\star$ is such that $\tilde{p}(\type,v) < 1$; there is a positive measure of such $\type$ as it follows from Lemma~\ref{lem:asymmetrichelp}\ref{part:asymmetrichelpunique} that $\tilde{p}(\type,1-c_m(X^\star-\type))$ is continuous in $\type$ and satisfies $\tilde{p}(\gamma^\star,1-c_m(X^\star-\gamma^\star)) = 1-c_m(X^\star-\gamma^\star) < 1$.
Almost surely, we then have that $p^\star(X) = \tilde{p}(\type,v) < 1$.

\label{proof:prop:manipulation_bunching_equilibrium}


To prove 
Proposition~\ref{prop:manipulation_equilibrium},
suppose that $\type \in (\gamma^\star,X^\star)$.
If $\tilde{p}(\type,v) < 1$, then by Lemma~\ref{lem:asymmetrichelp}\ref{part:asymmetrichelpunique}, we have
\begin{align*}
c_m \beta_\Delta(\type,v) &= \tilde{p}(\type,v) - v = \frac{2v + \sqrt{v^2 + 3 c_m^2 (\type^2 - (\gamma^\star)^2)}}{3} - v > v - v = 0.
\end{align*}
On the other hand, if $\tilde{p}(\type,v) = 1$, then  then by Lemma~\ref{lem:asymmetrichelp}\ref{part:asymmetrichelpunique}, we have $\beta_\Delta(\type,v) = 1-v > 0$.
Hence, almost surely, we have that $\beta_{\Delta^\star_{p^\star}} > 0$.

To prove
Proposition~\ref{prop:bunching_equilibrium},
note that 
at $v = 1$ and $\type = X^\star$, we have
$$
\frac{2 v + \sqrt{v^2 + 3 c_m^2(\type^2 - \gamma^{\star 2})}}{3} = 
\frac{2  + \sqrt{1 + 3 c_m^2(X^{\star 2} - \gamma^{\star 2})}}{3} \ge \frac{7}{3}  > 1
$$
Hence, by continuity, there exists $\zeta > 0$ such that for all $\type \in (X^\star - \zeta,X^\star)$, we have
\[\frac{2v + \sqrt{v^2 + 3 c_m^2 (\type^2 - (\gamma^\star)^2)}}{3} > 1.\]
For such $\type$, by Lemma~\ref{lem:asymmetrichelp}\ref{part:asymmetrichelpunique}, we have $\tilde{p}(\type,v) = 1$ and $\type + \beta_\Delta(\type,v) = \type + \frac{1-v}{c_m} = X^\star$.
Hence, almost surely, we have that $\type + \beta_{\Delta^\star_{p^\star}} = X^\star$.

\section{Proof of Proposition \ref{prop:comparison_pap} } \label{proof:pre_analysis}

\paragraph*{Proof of (a)}
In the notation of Appendix~\ref{app:proofsPhacking}, by Lemma~\ref{lem:lossThmAsymmetric}, we have
\begin{align*}
\mathcal{L}^\star_M &\ge \mathbb{E}_{\type}\left[\min_{\Delta,p} \left\{\big(\omega^2(\beta_\Delta^2-\type^2)+c_a\big)p + \omega^2 \type^2+\omega S^2 \right\}\right]\\
&= \mathbb{E}_{\type}\left[\min_{p} \left\{(c_a - \omega^2 \type^2)p + \omega^2 \type^2+ \omega S^2 \right\}\right]\\
&= \mathbb{E}_{\type}\Big[(c_a - \omega^2 \type^2) 1\{|\type| \ge \gamma^\star\} + \omega^2 \type^2+\omega S^2 \Big],
\end{align*}
where the last equality holds as $\gamma^\star = \frac{\sqrt{c_a}}{\omega}$ (because $\omega = \frac{\eta^2}{\eta^2 + S^2}$).
It follows from Proposition~\ref{prop:manipulation_equilibrium} that the inequality in the equation above must be strict.
Hence, considering the publication rule $p^\star(X) = 1\{|X| \ge \gamma^\star\}$, by Lemma~\ref{lem:loss}, it follows that
\begin{align*}
\mathcal{L}^\star_M &> \mathbb{E}_{X(E)}\bigg[\mathbb{E}\Big[\mathcal{L}_{p^\star}\big(X(E),E,\theta\big)\Big|X(E)\Big]\bigg] = \mathcal{L}^\star_E,
\end{align*}
where the last equality holds by Corollary~\ref{cor:t_testCheapExpensive}(a) as $E$ is cheap.

\paragraph*{Proof of (b)} Consider the publication rule defined by $p(X) = 1\big\{|X| \ge \gamma^\star_E + \frac{1}{c_m}\big\}$.
By construction, writing $\Delta^\star_{p}$ for a best response of the researcher under manipulation, we have
\begin{equation*} 
\mathcal{L}_M^\star \le \mathbb{E}_{\theta,\varepsilon}\Big[\mathcal{L}_{p}\big(X(\Delta_{p}^\star), \Delta_{p}^\star, \theta\big)\Big].
\end{equation*}
Expanding the expression on the right-hand side and writing $X(\Delta_{p}^\star)$ yields that
\begin{align*}
\mathcal{L}_M^\star & \le \mathbb{E}_{\theta,\varepsilon}\left[\Big(\varepsilon + \beta_{\Delta_{p}^\star}\Big)^2 p(X) + \eta^2\big(1 - p(X)\big) + c_a p(X)\right] \\ 
& =\underbrace{\mathbb{E}_{\theta,\varepsilon}\Big[ \varepsilon^2 p(X) + \eta^2\big(1 - p(X)\big) + c_a p(X)\Big]}_{(A)} + \underbrace{\mathbb{E}_{\theta,\varepsilon}\Big[\beta_{\Delta_{p}^\star}^2 p(X)\Big] + 2 \mathbb{E}_{\theta,\varepsilon}\Big[\beta_{\Delta_{p}^\star} \varepsilon p(X)\Big]}_{(B)}
\end{align*}

Under $p$, all researcher types with $|\theta + \varepsilon| \in \big[\gamma_E^\star, \gamma_E^\star + \frac{1}{c_m}\big]$ choose $|\beta_{\Delta_{p}^\star}| = \gamma_E^\star + \frac{1}{c_m} - |\theta + \varepsilon|$. All other types choose $|\beta_{\Delta_{p}^\star}| = 0$. In particular, we have that $p(X(\Delta_{p}^\star)) = 1\{|\theta + \varepsilon| \ge \gamma_E^\star\}$. Hence, writing $\hat{p}(X) = 1\{|X| \ge \gamma_E^\star\}$, we have that
$$
\begin{aligned}
(A) & =  \mathbb{E}_{\theta,\varepsilon}\Big[ \varepsilon^2 p(X) + \eta^2\big(1 - p(X)\big) + c_a p(X)\Big] = \mathbb{E}_{\theta,\varepsilon}\Big[\mathcal{L}_{\hat{p}}(\theta + \varepsilon, 0, \theta)\Big] \\ 
(B) & = \mathbb{E}_{\theta,\varepsilon}\Big[\beta_{\Delta_{p}^\star}^2 \hat{p}(\theta + \varepsilon)\Big] + 2 \mathbb{E}\Big[\beta_{\Delta_{p}^\star} \varepsilon \hat{p}(\theta + \varepsilon)\Big]. 
\end{aligned} 
$$
That is $(A)$ equals the loss function as if there was no manipulation, under a publication rule $\hat{p}$. $(B)$ captures the manipulation component.

To bound $(B)$,
note that the researcher will never choose a bias larger than $1/c_m$ under publication rule $p$, we have $\beta_{\Delta_{p}^\star}^2 \le \frac{1}{c_m^2}.$
Hence, for the first term of $(B)$, we have
$$
\mathbb{E}_{\theta,\varepsilon}\Big[\beta_{\Delta_{p}^\star}^2 \hat{p}(\theta + \varepsilon)\Big] \le \frac{1}{c_m^2}
$$ 
For the second term of $(B)$, by the Cauchy-Schwarz inequality, we have
$$
\mathbb{E}_{\theta,\varepsilon}\Big[\beta_{\Delta_{p}^\star} \varepsilon \hat{p}(\theta + \varepsilon)\Big] \le \sqrt{\mathbb{E}_{\theta,\varepsilon}[\beta_{\Delta_{p}^\star}^2] S^2} \le \frac{S}{c_m}.
$$

Collecting the terms, we have that
$$
\mathcal{L}_E^\star - \mathcal{L}_M^\star \ge \mathcal{L}_E^\star - \mathbb{E}\Big[\mathcal{L}_{\hat{p}}(\theta + \varepsilon, 0, \theta)\Big]  - \frac{1 + 2 S c_m}{c_m^2}. 
$$ 
On the other hand, $\mathbb{E}\Big[\mathcal{L}_{\hat{p}}(\theta + \varepsilon, 0, \theta)\Big] = \mathcal{L}_{E'}^\star$ for a design $E'$ with cost $C_{E'} = 0$ and variance $S_{E'}^2 = S^2$. It follows that $\mathcal{L}_E^\star - \mathbb{E}\Big[\mathcal{L}_{\hat{p}}(\theta + \varepsilon, 0, \theta)\Big] = \mathrm{IC}(E)$,
and hence we have
\[\mathcal{L}_E^\star - \mathcal{L}_M^\star \ge \mathrm{IC}(E) - \frac{1 + 2 S c_m}{c_m^2} > 0.\]

\section{Additional figure}
\label{app:figs}

\begin{figure}[!ht]
\centering 
\includegraphics[scale = 0.5]{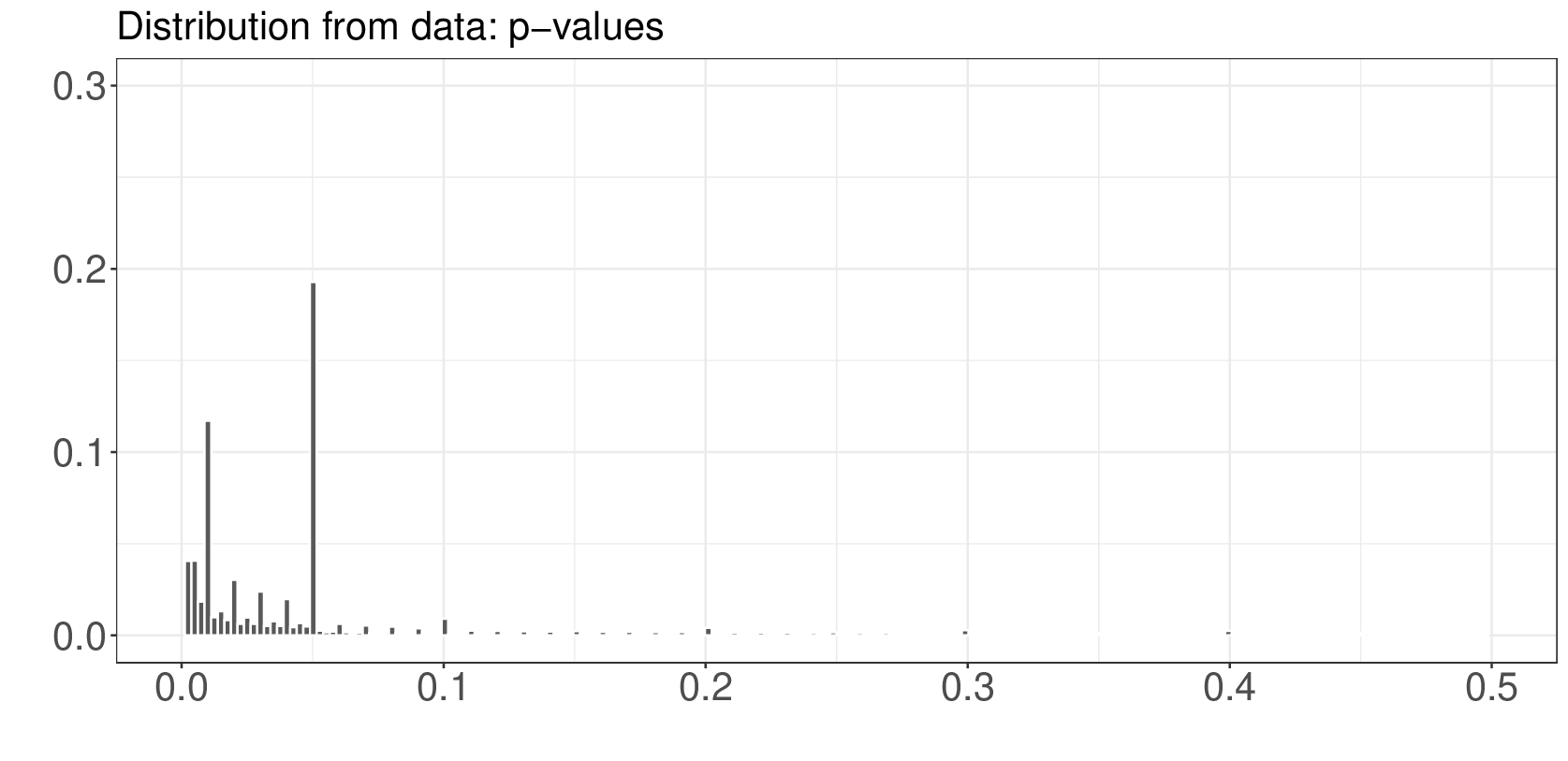}
\caption{Distribution of $p$-values for medical and pharmaceutical articles in \cite{head2015extent}.} \label{fig:raw}
\end{figure}

\end{document}